\input{\jobname.cfg}
\newcommand{\acmart}{\True}

\documentclass[acmsmall,screen]{acmart}
\setcopyright{rightsretainedccbysa}
\acmPrice{}
\acmDOI{}
\acmYear{2021}
\copyrightyear{2021}
\acmSubmissionID{popl21main-p349-p}
\acmJournal{PACMPL}
\acmVolume{5}
\acmNumber{POPL}
\acmArticle{45}
\acmMonth{1}

\acmArticleComment{Author version with appendices.}

\usepackage{reversion}

\usepackage[utf8]{inputenc}
\acmart{}{
\usepackage{libertine}
\usepackage{inconsolata}
}

\usepackage{mylistings}
\usepackage{mycomments}
\usepackage{mymath}
\usepackage{mybiblio}
\usepackage{myhyperref}
\usepackage{notations}
\usepackage{myrefs}

\Comments{}{\UNXXX}

\newcommand{\PR}[1]{\href{https://github.com/ocaml/ocaml/issues/#1}{\##1}}
\newcommand{\Fsharp}{F$\sharp$}

\newcommand{\longfootnote}[1]{\Long{\footnote{#1}}{\stepcounter{footnote}}}

\title{A Practical Mode System for Recursive Definitions}

\author{Alban Reynaud}
\affiliation{
  \institution{ENS Lyon}
  \country{France}
}

\author{Gabriel Scherer}
\affiliation{
  \institution{INRIA}
  \country{France}
}

\author{Jeremy Yallop}
\affiliation{
  \institution{University of Cambridge}
  \country{United Kingdom}
}

\begin{document}

\Draft{\hfill\textbf{Draft, \today}}{}


\begin{abstract}
In call-by-value languages, some mutually-recursive definitions
can be safely evaluated to build recursive functions or cyclic data
structures, but some definitions (\code{let rec x = x + 1}) contain
vicious circles and their evaluation fails at runtime. We propose
a new static analysis to check the absence of such runtime failures.

We present a set of declarative inference rules, prove its soundness
with respect to the reference source-level semantics of
\citet*{dynamic-semantics-2008}, and show that it can be
directed into an algorithmic backwards analysis check in a surprisingly
simple way.

Our implementation of this new check replaced the existing check used
by the OCaml programming language, a fragile syntactic
criterion which let several subtle bugs slip through as the language
kept evolving. We document some issues that arise when advanced
features of a real-world functional language (exceptions in
first-class modules, GADTs, etc.) interact with safety checking for
recursive definitions.
\end{abstract}

%% 2012 ACM Computing Classification System (CCS) concepts
%% Generate at 'http://dl.acm.org/ccs/ccs.cfm'.

%% End of generated code

\maketitle

\section{Introduction}

Recursion pervades functional programs.
Functional programmers often start out
by writing simple recursive definitions
such as the Fibonacci function,
shown here in OCaml:
\begin{lstlisting}
let rec fib = fun x -> if x <= 1 then x
                         $\;$else fib (x-1) + fib (x-2)
\end{lstlisting}
This definition is elegant but, alas, impractical:
computing \lstinline!fib n!
takes time exponential in \lstinline!n!.
One way to improve performance is to memoize:
in place of a function,
we might (recursively) define a lazy list, \lstinline{lfibs},
whose $n$th element represents \lstinline!fib n!:
\begin{lstlisting}
let rec lfibs = lazy (0 :: lazy (1 :: map2 (+) lfibs (tail lfibs)))
\end{lstlisting}
%
%% However, memoization through lazy lists is rarely used in eager
%% languages such as OCaml, where elegance suffers from the need to make
%% all laziness explicit.  (For example, \lstinline{lfibs} cannot use the
%% standard library's \lstinline{map2}; it needs a variant that consumes
%% and produces lazy lists.)

%% In eager languages programmers typically define function and memo-table
%% separately.  The \lstinline{memofun} type pairs a
%% function \lstinline{f} with a table of its previously computed values:
%% %
%% \begin{lstlisting}
%% type ('k,'v) memofun = { f: 'k -> 'v;  values: ('k, 'v) hashtable }
%% \end{lstlisting}
%% %
%% The main interface to \lstinline{memofun} is the \emph{apply}-like function \lstinline{remember}:
%% %
%% \begin{lstlisting}
%% val remember : ('k, 'v) memofun -> 'k -> 'v
%% \end{lstlisting}
%% %
%% In English, ``remember'' can mean
%% either \emph{retrieve from memory}
%% or \emph{store in memory}.
%% %
%% Similarly,
%% \lstinline!remember!
%% either retrieves a stored result for its argument,
%% or computes and stores a new result:
%% %
%% \begin{lstlisting}
%% let remember t k = try find t.values k
%%                      $\,$with Not_found -> let v = t.f k in add t.values k v; v
%% \end{lstlisting}

As these definitions show, recursion is useful for defining both
functions and values.
However, lazy-list memoization is rarely used in eager languages,
since elegance suffers from the need to make all laziness explicit.
Here is a more idiomatic memoized function \lstinline{mfib},
mutually-defined with a record \lstinline{mfibs},
which pairs the function with a memo-table of its previously-computed values:
\begin{lstlisting}
let rec mfib = fun x -> if x <= 1 then x
                          $\;$else remember mfibs (x-1) + remember mfibs (x-2)
    $\,$and mfibs = { f = mfib; values = empty_table () }
\end{lstlisting}
(The \lstinline!remember! function retrieves previously-computed
values and computes and stores new entries.)
This definition exposes
both \lstinline{mfib} and its table \lstinline{mfibs},
risking inadvertent modification of the table.
A cautious programmer might avoid this danger by making \lstinline{mfibs} local to \lstinline{mfib}:
\begin{lstlisting}
let rec mfib' = let mfibs' = { f = mfib'; values = empty_table () } in
                 $\;$fun x -> if x <= 1 then x
                           $\;$else remember mfibs' (x-1) + remember mfibs' (x-2)
\end{lstlisting}

\subsection{Recursion: Expressiveness vs Safety}

As this tapestry of \lstinline{fib}s suggests, the \emph{usefulness} of
recursion is not limited to simple function definitions: our examples
build record values, call functions and bind local names.
However, in eager languages not all recursion is \emph{safe}.
For example, here is an unsafe eager variant of \lstinline{lfibs}:
\begin{lstlisting}
let rec efibs = 0 :: 1 :: map2 (+) efibs (tail efibs)   (* unsafe! *)
\end{lstlisting}
Evaluation of \lstinline{efibs} fails, since
\lstinline{map2} and \lstinline{tail} access \lstinline{efibs} 
before it is fully constructed.

Existing functional languages incorporate various approaches to
balancing usefulness and safety.

%% (Haskell evaluates lazily, typically loops on cycles)

In some languages, such as Scheme~\citep{r6rs}, mutually-recursive
bindings are evaluated immediately, and it is a run-time error for
evaluation to encounter any identifier being bound.
Evaluating \lstinline{mfibs}, which refers to \lstinline{mfib}, would
produce such an error.

Other languages, such as \Fsharp{}, provide a kind of lazy evaluation
for recursive value bindings~\citep{syme-2006} (discussed
in \S\ref{subsec:related-work}) which supports the construction of
recursive objects.
\Fsharp{}'s approach allows more programs to execute without error --- it is
sufficient to support \lstinline{mfib}, though not \lstinline{mfib'}
--- but it does not entirely eliminate run-time errors from ill-formed
recursive bindings.
Consequently, an additional syntactic check~\cite{syme2005alternative}
rejects some cases of self-reference that would result in
run-time errors.  (However, even this additional check is not
sufficient to eliminate \emph{all} such errors.)

Finally, some languages, such as Standard ML~\citep{sml97}, incorporate
a more severe approach, permitting recursion only through syntactic
function definitions such as \lstinline{fib},
and statically rejecting \lstinline{mfib}, \lstinline{mfib'} and \lstinline{efib}.
(Some implementations also support laziness, and
allow \lstinline{lfibs}.)

%% mfib is allowed 
%% mfib' is disallowed ("The value 'mfib'' will be evaluated as part of its own definition")
%% efibs is disallowed

Of these design choices, Standard ML's most fully embodies Milner's dictum:
\emph{well-typed programs do not go wrong}.
However, as we show in this paper, Standard ML's treatment of
recursive definitions is unnecessarily restrictive: it is possible
to define a much more permissive criterion for recursive definitions
that still ensures the absence of run-time errors.
Our criterion allows useful definitions such as \lstinline{lfibs},
\lstinline{mfib} and \lstinline{mfib'}, while rejecting incorrect
programs such as \lstinline{efibs}.  We present the criterion as a mode system
(proved sound with respect to an operational semantics), suitable for
incorporating into a compiler --- indeed, our implementation has
already been merged into the mainline OCaml compiler.

\subsection{OCaml Needed Fixing}

Before the work described in this paper, OCaml took an approach
similar to \Fsharp{}'s (although somewhat more precise, and based on
OCaml's existing eager semantics rather than a translation
into thunks), checking for vicious recursive definitions via syntactic
analysis of an intermediate representation of programs. We believe
OCaml's check as originally
defined\footnote{\url{https://ocaml.org/releases/4.05/htmlman/extn.html\#sec217}}
was correct, but it proved fragile and difficult to maintain as the
language evolved and new features interacted with recursive
definitions. Over the years, several bugs were found where the check
was unduly lenient (\S\ref{section:issues-with-the-previous-check}).
In conjunction with OCaml's efficient compilation scheme for recursive
definitions~\citep{hirschowitz-compilation-2009}, this leniency
resulted in memory safety violations, and led to segmentation faults
when definitions that accessed recursively-defined objects before they
were initialized were allowed through undetected.

\subsection{Generalized Recursion in Practice}

In order to determine whether generalized recursive definitions were
used in practice, we searched a large subset of packages in the OCaml
software repository (OPAM) for instances of generalized recursion
(i.e.~recursive definitions that Standard ML would reject).
At the time of our analysis in July 2020, we found 309 distinct
examples of such definitions in around 74 packages.
%
% Note: the 'hard' filter excludes OASIS usage, so we don't
% need to talk about OASIS explicitly anymore.
%
The definitions variously made use of local bindings and other local
constructs (such as local exception definitions),
evaluation of sub-terms, record and variant construction and laziness,
in both singly- and mutually-recursive bindings.
Generalized recursion definitions are used in useful data structures
and important libraries that many other packages in turn depend on. We
found that 1613 packages out of 2819 in the repository at the time
(57\%) depend on at least one of those packages using generalized
recursion.

\subsection{Our Analysis}

The present document formally describes our analysis using a core ML
language (\Sref{section:core-language}). We present inference rules
(\Sref{section:inference-rules}) and study the meta-theory of the
analysis. We propose a source-level operational semantics, refreshing
semantics proposed in earlier
works~\citep*{ariola-felleisen,dynamic-semantics-2008,hirschowitz-compilation-2009}
with explicit substitutions using \emph{reduction at a distance}
(\Sref{section:meta-theory}), and show that our analysis is sound for
this semantics. We also propose a semantics that uses mutable
updates to a global store, closer to production-compiler compilation
strategies (\Sref{section:global-store}), for which our analysis is also
sound. Finally, we discuss the challenges caused by scaling the
analysis to OCaml (\Sref{section:extension-full}), a full-fledged
functional language, in particular the delicate interactions with
non-uniform value representations (\Sref{section:float-arrays}), with
exceptions and first-class modules (\Sref{section:exceptions}), and
with Generalized Algebraic Datatypes (GADTs) (\Sref{section:gadts}).

\subsection{Adoption in OCaml}

The initial version of our new system was originally released in OCaml
4.06.0 (3 Nov 2017).
We formalised the new system after its release, and reworked the
implementation to better match the formalisation.  The updated
implementation was released in OCaml 4.08.0 (13 June 2019).

Before releasing our new system, we sought to determine whether it was
significantly more restrictive in practice than the previous check by
running it over the packages in OPAM.  This investigation did not
uncover any code that was accepted by the previous check but rejected
by our new system.
Six major OCaml releases later there has been only one report of such
a program (\PR{7767}); it was straightforward to update our system to
accept it.

Of course, since we looked only at existing OCaml code, our
investigation could not uncover definitions that were rejected by the
previous check but accepted by our new system, which is more
permissive in some cases
(\S\ref{section:issues-with-the-previous-check}).
However, our goal was not to increase the expressivity of OCaml's
value definitions, but to design a static check that was
backwards-compatible with the previous check, while being easier to
reason about and evolve in tandem with the
language~(\S\ref{section:extension-full}).
Our aim was to make the check as simple as possible within the
constraints, not as expressive as possible.

Furthermore, moving the check from the compiler middle end into
the type checker has another benefit: it is convenient for tools
that reuse OCaml's type-checker without performing compilation, such
as the multi-stage language
MetaOCaml~\cite{DBLP:conf/flops/Kiselyov14} (which type-checks code
quotations) and the Merlin language
server~\cite{DBLP:journals/pacmpl/BourRS18} (which type-checks code
during editing).

\subsection{Contributions}

We claim the following contributions:
\begin{itemize}
\item We propose a new system of inference rules that captures the
  safety conditions for recursive definitions in an eager language
  (\Sref{section:inference-rules}), previously enforced in OCaml
  by ad-hoc syntactic restrictions.
\item We prove the analysis sound with respect to a source-level
  operational semantics: accepted recursive terms evaluate without
  vicious-circle failures (\Sref{section:meta-theory}).  Our
  source-level semantics is justified by a simulation result with
  lower-level backpatching semantics with a global store
  (\Sref{section:global-store}), for which our analysis is also sound.
\item We have implemented a checker derived from the rules, scaled up to
  the full OCaml language (\S\ref{section:extension-full}) and
  integrated in the OCaml implementation.
\item Our analysis is less fine-grained on functions than existing
  works (\S\ref{subsec:related-work}), thanks to a less demanding
  problem domain (ML functions rather than ML functors), but in exchange
  it provides finer-grained handling of cyclic data and an effective
  inference algorithm.
\end{itemize}

\section{Overview}

\subsection{Access Modes}
\label{subsec:intro-modes}

Our analysis is based on the classification of each use of a
recursively-defined variable using ``access modes'' or ``usage modes''
$m$.  These modes represent the degree of access needed to the value
bound to the variable during evaluation of the recursive definition.

For example, in the recursive function definition
\begin{lstlisting}
let rec f = fun x -> $\ldots$ f $\ldots$
\end{lstlisting}
the recursive reference to \code{f} in the right-hand-side does not
need to be evaluated to define the function value
\code{fun x -> $\ldots$}
since its value will only be required later, when the function is applied.
We say that, in this right-hand-side, the mode of use of the variable \code{f}
is $\Delay$.

In contrast, in the vicious definition
\code{let rec x = 1 + x}
evaluation of the right-hand side %% \code{1 + x}
involves accessing the
value of |x|; we call this usage mode a $\Dereference$. Our static
check rejects mutually-recursive definitions that access
recursively-bound names under this mode.

Some patterns of access fall between the extremes of $\Delay$ and
$\Dereference$.  For example, in the cyclic datatype construction
\lstinline!let rec obj = { self = obj }!
the recursively-bound variable \code{obj} appears on the right-hand side
without being placed inside a function abstraction.  However, since it
appears in a ``guarded'' position, within the record constructor |{self = -}|,
evaluation only needs to access its address, not its
value. We say that the mode of use of the variable \code{ones} is
$\Guard$.

Finally, a variable |x| may also appear in a position where its value
is not inspected, neither is it guarded beneath a constructor, as in
the expression \code{x}, or \code{let y = x in y}, for example. In
such cases we say that the value is ``returned'' directly and use the
mode $\Return$. As with |Dereference|, recursive definitions that
access variables at the mode $\Return$, such as %% the following
\code{let rec x = x},
would be under-determined and are rejected.

We also use a last $\Ignore$ mode to classify variables that are not
used at all in a term.

\subsection{An Inference System (and Corresponding Backwards Analysis)}
\label{subsec:intro-right-to-left}
\label{subsec:intro-rules}

The central contribution of our work is a simple system of inference
rules for a judgment of the form $\der \Gamma t m$, where $t$ is
a program term, $m$ is an access mode, and the environment $\Gamma$
maps term variables to access modes. Modes classify terms and
variables, playing the role of types in usual type systems. The
example judgment
%
% {%leave a blank space above otherwise we hit a bug in LIPIcs line-counting logic
% \abovedisplayskip -1em
% \belowdisplayskip 0em
% \begin{mathpar}
$
\der {x : \Dereference, y : \Delay} {(x+1, \lazy{y})} \Guard
$
% \end{mathpar}
% }
can be read alternatively
\begin{description}
\item[forwards:] If we know that $x$ can safely be used in
  $\Dereference$ mode, and $y$ can safely be used in $\Delay$ mode,
  then the pair $(x+1, \lazy{y})$ can safely be used under a value
  constructor (in a $\Guard$-ed context).
\item[backwards:] If a context accesses the program fragment
  $(x+1, \lazy{y})$ under the mode $\Guard$, then this means that the
  variable $x$ is accessed at the mode $\Dereference$, and the variable $y$
  at the mode $\Delay$.
\end{description}
\hspace{\parindent}
This judgment uses access modes for two purposes: to classify
variables, and to classify the constraints imposed on a subterm by its
surrounding context. If a context $C[\hole]$ uses its hole $\hole$ at
the mode $m$, then any derivation for the plugged context $C[t]
: \Return$ will contain a sub-derivation of the form $t : m$ for the
term $t$.

In general, we can define a notion of mode composition: if we try to
prove $C[t] : m'$, then the sub-derivation will check
$t : \mcomp {m'} m$, where $\mcomp {m'} m$ is the composition of the
access-mode $m$ under a surrounding usage mode $m'$, and $\Return$ is
neutral for composition.

Our judgment $\der \Gamma t m$ can be directed into an algorithm
following our backwards interpretation. Given a term $t$ and a
mode $m$ as inputs, our algorithm computes the least demanding environment
$\Gamma$ such that $\der \Gamma t m$ holds.

For example, the inference rule for function abstractions in our
system is as follows:
\begin{mathpar}
  \infer
  {\der {\Gamma, \of x {m_x}} t {\mcomp m \Delay}}
  {\der \Gamma {\lam x t} m}
\end{mathpar}
The backwards reading of the rule is as follows. To compute the
constraints $\Gamma$ on $\lam x t$ in a context of mode $m$, it
suffices to check the function body $t$ under the weaker mode
$\mcomp m \Delay$, and remove the function variable $x$ from the
collected constraints --- its mode does not matter. If $t$ is
a variable $y$ and $m$ is $\Return$, we get the environment
$\of y \Delay$ as a result.

Given a family of mutually-recursive definitions
$\letrecfam {i \in I} {x_i} {t_i}$, we run our algorithm on each $t_i$
at the mode $\Return$, and obtain a family of environments
$\fam{i \in I}{\Gamma_i}$ such that all the judgments
$\fam{i \in I}{\der {\Gamma_i} {t_i} \Return}$ hold. The definitions
are rejected if one of the $\Gamma_i$ contains one of the
mutually-defined names $x_j$ under the mode $\Dereference$ or
$\Return$ rather than $\Guard$ or $\Delay$.

\subsection{Issues with the Previous Check}
\label{section:issues-with-the-previous-check}

 Before this work, the
safety criterion used by OCaml for recursive value definitions was an
ad-hoc grammatical restriction, formulated essentially as
a context-free grammar of accepted definitions (see its description in
\href{http://caml.inria.fr/pub/docs/manual-ocaml-4.09/manual023.html}{the
  reference manual}). Furthermore, this syntactic check was not
performed on the source program directly, but on an intermediate
representation (the Lambda code) --- so that it wouldn't have to take
into account various surface-language forms that desugar to the same
intermediate-language construct.

We list below some of the known issues with the previous check. They
were solved by our work.

\paragraph{\PR{7231}: unsoundness with nested recursive bindings} The
previous check accepted the following unsafe program.
\begin{lstlisting}
let rec r = let rec x () = r
                  and y () = x ()
             in y ()
       in r "oops" (* segfault *)
\end{lstlisting}
The problem is that while the declarations of \lstinline{x} and
\lstinline{y} are ``safe'' (in some sense) with respect to
\lstinline{r}, using \lstinline{y} is not safe --- it returns
\lstinline{r} itself. This subtlety was lost on the previous
check. With the current check, \lstinline{y ()} uses \lstinline{r} at
mode \lstinline{Return}, which is stricter than $\Guard$, so
this program is rejected.

\paragraph{\PR{7215}: unsoundness with GADTs} The previous check accepted the following unsafe
program.
\begin{lstlisting}
let is_int (type a) =
  let rec (p : (int, a) eq) = match p with Refl -> Refl in p
\end{lstlisting}
This program uses a recursive value declaration of a GADT value to
build a type-equality between \lstinline{int} and an arbitrary type
\lstinline{a}. Our check rejects the program because
\lstinline{match p with Refl -> ...}
is a dereferencing use of \lstinline{p}. The previous check was run on
an intermediate form, after various optimizations, one of which would
eliminate the single-case match away, resulting in the (unsound)
program passing the check.

\paragraph{\PR{6939}: unsoundness with float arrays}$~$
\begin{lstlisting}
let rec x = ([| x |]; 1.) in ()
\end{lstlisting}
This program defines \lstinline{x} to be the floating-point value
\lstinline{1.} after ignoring the value of the one-element array
\lstinline{[| x |]}. Although the program was accepted by the previous check,
OCaml's non-uniform value representation makes it unsafe, and it would
fail with a segmentation fault when run, as explained in
\myfullref{Section}{section:float-arrays}. Our algorithm uses
typing information, which is needed to detect this case: construction
of a float array is treated as a $\Dereference$ context for its
elements.

\paragraph{\PR{4989}: inconveniently rejected program}$~$
\begin{lstlisting}
let rec f = let g = fun x -> f x in g
\end{lstlisting}
This program, which gives a local name to an expression
that accesses \lstinline{f} at mode $\Delay$, is perfectly safe, but
was rejected by the previous check.
A grammar-based check lacks a form of composability
that would allow the use of local bindings to give names to
sub-expressions in an analysis-preserving way. This issue (\PR{4989}) dates back to 2010: this form of
composability had been requested by users for a long time as
a convenience feature, but the previous check could not be extended to
allow it. On the contrary, proper handling of inner
\lstinline{let}-bindings falls out naturally from our type-system-inspired
approach.

\section{A Core Language of Recursive Definitions}
\label{section:core-language}

\paragraph{Family notation} We write $\fam {i \in I} {\dots}$ for
a family of objects parametrized by an index $i$ over finite set $I$,
and $\emptyset$ for the empty family. Furthermore, we assume that
index sets are totally ordered, so that the elements of the family are
traversed in a predetermined linear order; we write
$\fam {i_1 \in I_1} {t_{i_1}}, \fam{i_2 \in I_2} {t_{i_2}}$ for the
combined family over $I_1 \uplus I_2$, with the indices in $I_1$
ordered before the indices of $I_2$. We often omit the index set,
writing $\fam i {\dots}$. Families may range over two indices (the
domain is the cartesian product), for example $\fam {i,j} {t_{i,j}}$.

Our syntax, judgments, and inference rules will often use families:
for example, $\letrecfam i {x_i} {t_i}$ is a mutually-recursive
definition of families $\fam i {t_i}$ of terms bound to corresponding
variables $\fam i {x_i}$ --- assumed distinct, following the Barendregt
convention. Sometimes a family is used where a term is expected, and
the interpretation should be clear: when we say
``$\fam i {\der {\Gamma_i} {t_i} {m_i}}$ holds'', we implicitly use
a conjunctive interpretation: each of the judgments in the family
holds.

\subsection{Syntax}
\begin{mathparfig}[t!]{fig:syntax}{Core language syntax}
  \begin{array}{l@{~}r@{~}l}
    \set{Terms} \ni t, u & \bnfeq & x, y, z \\
    & \bnfor & \letrecin b u \\
    & \bnfor & \lam x t %  \\
    % &
      \bnfor % &
               \app t u \\
    & \bnfor & \constr K {\fam i {t_i}} % \\
    % &
      \bnfor % &
               \match t h \\ \\
  \end{array}

  \begin{array}{l@{~}r@{~}l}
    \set{Bindings} \ni b & \bnfeq & {\fam i {x_i = t_i}} \\
    % \\
    \set{Handlers} \ni h & \bnfeq & {\fam i {\clause {p_i} {t_i}}} \\
    % \\
    \set{Patterns} \ni p, q & \bnfeq & \constr K {\fam i {x_i}} \\
  \end{array}
\end{mathparfig}

\myref{Figure}{fig:syntax} introduces a minimal subset of ML containing
the interesting ingredients of OCaml's recursive values:
\begin{itemize}
\item A multi-ary \code{let rec} binding $\letrecin {\fam i {x_i = t_i}} u$.
\item Functions ($\lambda$-abstractions) $\lam x t$ to write recursive
  occurrences whose evaluation is delayed.
  %% (OCaml has additional
  %% constructs for delaying computation, such as |lazy| values and
  %% object literals.)
\item Datatype constructors $\constr K (t_1, t_2, \dots)$ to write
  (safe) cyclic data structures; these stand in both for user-defined
  constructors and for built-in types such as lists and tuples.
\item Shallow pattern-matching
  $(\match t {\fam i {\clause {\constr {K_i} {\fam j {x_{i,j}}}} {u_i}}})$,
  to write code that inspects values, in particular code with vicious
  circles.
\end{itemize}

The following common ML constructs do not need to be primitive forms,
as we can desugar them into our core language. In particular, the full
inference rules for OCaml (and our check) exactly correspond to the
rules (and check) derived from this desugaring.
%% \begin{itemize}
%% \item $n$-ary tuples are a special case of constructors:\\
%%   $(t_1, t_2, \dots, t_n)$ desugars into
%%   $\constr {\Tuple_n} {\fam {i \in [1;n]} {t_i}}$.
%% \item Non-recursive \emph{let} bindings are recursive bindings with
%%   access mode $\Ignore$:\\ $\letin x t u$ desugars into
%%   $\letrecin {x = t} u$.
%% \item Conditionals are a special case of pattern-matching:\\
%%   $\ifte t {u_1} {u_2}$ desugars into
%%   $\match t {(\clause {\kwd{True}} {u_1} \mid \clause {\kwd{False}} {u_2})}$.
%%  % \item Sequencing: $\seq t u$ desugars into
%% %   $\match t (\clause \wild u)$.  \Xgabriel{Due to the discarding
%% %     issue, $t$ is $\Return$ in this desugaring, not $\Guard$.}
%% \end{itemize}

Besides dispensing with many constructs whose essence is captured by
our minimal set, we further simplify matters by using an untyped ML
fragment: we do not need to talk about ML types to express our check,
or to assume that the terms we are working with are
well-typed.\footnote{In more expressive settings, the structure of
  usage modes does depend on the structure of values, and checks need
  to be presented as a refinement of a ML type system. We discuss this
  in \myref{Section}{subsec:related-work}. Our modes are
  a degenerate case, a refinement of uni-typed ML.}
However, we do assume that our terms are well-scoped --- note
that, in $\letrecin {\fam i {x_i = v_i}} u$, the $\fam i {x_i}$ are in
scope of $u$ but also of all the $v_i$.

\begin{remark} Recursive values are a controversial feature as they
  break the assumption that structurally-decreasing recursive
  functions will terminate on all inputs. The uses we found in the
  wild in OCaml programs typically combine ``negative'' constructs
  (functions, lazy, records) rather than infinite lists or
  trees. A possible design would be to distinguish an ``inductive''
  sub-space of recursive types whose recursive occurrences are
  forbidden in negative positions, and whose constructors are not
  given the $\Guard$ mode in our system. In another direction,
  \citet*{cocaml} propose language extensions to make it easier to
  operate over cyclic structures.
\end{remark}

%% \begin{remark}
%%   Recursive values are a controversial language feature; there are
%%   sensible arguments against them. Our work does not take a stance on
%%   whether the feature should be added to new languages, but
%%   demonstrates how to safely and robustly support it in languages that
%%   decide to have it.
%% \end{remark}

\section{A Mode System for Recursive Definitions}
\label{section:inference-rules}

\subsection{Access/Usage Modes}
\label{subsec:modes}
\begin{mathparfig}[t!]{fig:modes}{Access/usage modes and operations}
  \begin{array}{l@{~}r@{~}l}
    \set{Modes} \ni m & \bnfeq & \Ignore \\
    & \bnfor & \Delay \\
    & \bnfor & \Guard \\
    & \bnfor & \Return \\
    & \bnfor & \Dereference \\
  \end{array}

\begin{array}{l} \text{Mode order:} \\
\Ignore \prec \Delay \prec \Guard \prec \Return \prec \Dereference
\end{array}

\begin{array}{l}
\text{Mode composition rules}: \\
  \begin{array}{lll}
    \mcomp \Ignore m & = & \Ignore \quad = \quad \mcomp m \Ignore
    \\
    \mcomp \Delay {m \succ \Ignore} & = & \Delay
    \\
    \mcomp \Guard \Return & = & \Guard
    \\
    \mcomp \Guard {m \neq \Return} & = & m
    \\
    \mcomp \Return m & = & m
    \\
    \mcomp \Dereference {m \succ \Ignore} & = & \Dereference
    \\
\end{array}

\\
\\

\text{Mode composition as a table:} \\
\begin{array}{lccccccc}
\mcomp m {m'} & \Ignore & \Delay & \Guard       & \Return      & \Dereference & m \\
\hline
\Ignore      & \Ignore & \Ignore & \Ignore      & \Ignore      & \Ignore      \\
\Delay       & \Ignore & \Delay  & \Delay       & \Delay       & \Dereference \\
\Guard       & \Ignore & \Delay  & \Guard       & \Guard       & \Dereference \\
\Return      & \Ignore & \Delay  & \Guard       & \Return      & \Dereference \\
\Dereference & \Ignore & \Delay  & \Dereference & \Dereference & \Dereference \\
m' &
\end{array}
\end{array}
\end{mathparfig}

\myref{Figure}{fig:modes} defines the access/usage modes that we
introduced in \myref{Section}{subsec:intro-modes}, their order
structure, and the mode composition operations.
The modes are as follows: \Xgabriel{Some items below explain both
  views ``mode of a subterm'' and ``mode of a variable'', it would be
  nice to consistently use this style for all modes.}
\begin{description}
\item[$\Ignore$] is for sub-expressions that are not used at all during
  the evaluation of the whole program. This is the mode of a variable
  in an expression in which it does not occur.
\item[$\Delay$] means that the context can be evaluated (to a weak
  normal-form) without evaluating its argument. $\lam x \hole$ is
  a delay context.
\item [$\Guard$] means that the context returns the value as a member
  of a data structure, for example a variant constructor or
  record. $\constr K {(\hole)}$ is a guard context. The value can
  safely be defined mutually-recursively with its context, as in
  $\letrec {x = {\constr K {(x)}}}$.
\item [$\Return$] means that the context returns its value without
  further inspection. This value cannot be defined
  mutually-recursively with its context, to avoid self-loops: in
  $\letrec {x = x}$ and $\letrec {x = \letin y x y}$, the rightmost
  occurrence of $x$ is in $\Return$ context.
\item [$\Dereference$] means that the context consumes, inspects and
  uses the value in arbitrary ways. Such a value must be fully defined
  at the point of usage; it cannot be defined mutually-recursively
  with its context. $\match \hole h$ is a $\Dereference$ context.
\end{description}

\begin{remark}[Discarding]\label{rem:discarding}
  The $\Guard$ mode is also used for subterms whose result is
  discarded by the evaluation of their context. For example, the hole
  $\hole$ is in a $\Guard$ context in $(\letin x \hole u)$, if $x$ is
  never used in $u$; even if the hole value is not needed,
  call-by-value reduction will first evaluate it and discard it. When
  these subterms participate in a cyclic definition, they cannot
  create a self-loop, so we consider them guarded.
\end{remark}

Our ordering $m \prec m'$ places less demanding, more permissive modes
that do not involve dereferencing variables (and so permit their use
in recursive definitions), below more demanding, less permissive
modes.

Each mode is closely associated with particular expression contexts.
For example, $\app t \hole$ is a $\Dereference$ context, since the function
$t$ may access its argument in arbitrary ways, while $\lam x \hole$ is a
$\Delay$ context.

Mode composition corresponds to context composition, in the sense that
if an expression context $\plug E \hole$ uses its hole at mode $m$
(to compute a result), and a second expression context
$\plug {E'} \hole$ uses its hole at mode $m'$, then the composition of
contexts $\plug E {\plug {E'} \hole}$ uses its hole at mode
$\mcomp m {m'}$.
Like context composition, mode composition is associative, but not
commutative: $\mcomp \Dereference \Delay$ is
$\Dereference$, but $\mcomp \Delay \Dereference$ is $\Delay$.

% JY: I've removed "for instance" after "not commutative" because
%     this is the only example of non-commutativity.
%
%     In fact, there are relatively few interesting combinations:
%       * Ignore is a left/right annihilator/zero
%       * Return is a left/right identity
%       * modes are idempotent: m[m] == m for each m
%     That leaves 6 combinations
%        Guard[Delay] = Delay       Delay[Guard] = Delay
%        Guard[Deref] = Deref       Deref[Guard] = Deref
%        Delay[Deref] = Delay       Deref[Delay] = Deref

Continuing the example above, the context $\app t {(\lam x \hole)}$,
formed by composing $\app t \hole$ and $\lam x \hole$, is
a $\Dereference$ context: the intuition is that the function $t$ may
pass an argument to its input and then access the result in arbitrary
ways.
In contrast, the context $\lam x {(\app t \hole)}$,
formed by composing $\lam x \hole$ and $\app t \hole$,
is a $\Delay$ context: the contents of the hole will
not be touched before the abstraction is applied.

Finally, $\Ignore$ is the absorbing element of mode composition
($\mcomp m \Ignore = \Ignore = \mcomp \Ignore m$), $\Return$ is
an identity ($\mcomp \Return m = m = \mcomp m \Return$),
and composition is idempotent ($\mcomp m m = m$).

\subsection{Inference Rules}

\paragraph{Environment notations} Our environments $\Gamma$ associate
variables $x$ with modes $m$. We write $\Gamma_1, \Gamma_2$ for the
union of two environments with disjoint domains, and
$\envsum {\Gamma_1} {\Gamma_2}$ for the merge of two overlapping
environments, taking the maximum mode for each variable. We sometimes
use family notation for environments, writing  $\fam i {\Gamma_i}$
to indicate the disjoint union of the members,
and $\envbigsum {\fam i {\Gamma_i}}$
for the non-disjoint merge of a family of environments.

\begin{mathparfig}[t!]{fig:rules}{Mode inference rules}
\fbox{Term judgment $\der \Gamma t m$}
  \\
  \infer{ }
  {\der {\Gamma, \of x m} x m}

  \infer
  {\der \Gamma t m\\
   m \succ m'}
  {\der \Gamma t {m'}}
  \\
  \infer
  {\der {\Gamma, \of x {m_x}} t {\mcomp m \Delay}}
  {\der \Gamma {\lam x t} m}

  \infer
  {\der {\Gamma_t} t {\mcomp m \Dereference} \\
   \der {\Gamma_u} u {\mcomp m \Dereference}}
  {\der {\envsum {\Gamma_t} {\Gamma_u}} {\app t u} m}
  \\
  \infer
  {\fam i {\der {\Gamma_i} {t_i} {\mcomp m \Guard}}}
  {\der {\envbigsum {\fam i {\Gamma_i}}} {\constr K {\fam i {t_i}}} m}

  \infer
  {\der {\Gamma_t} t {\mcomp m \Dereference} \\
   \derclause {\Gamma_h} h m}
  {\der {\envsum {\Gamma_t} {\Gamma_h}} {\match t h} m}
  \\
  \infer
  {\derbinding {\fam i {\of {x_i} {\Gamma_i}}} b \\
   \fam i {m'_i} \defeq \fam i {\max(m_i, \Guard)} \\
   \der {\Gamma_u, \fam i {\of {x_i} {m_i}}} u m}
  {\der
   {\envsum
     {\envbigsum {\fam i {\mcomp {m'_i} {\Gamma_i}}}}
     {\Gamma_u}}
   {\letrecin b u}
   m}
 \\
  \fbox{Clause judgments $\derclause \Gamma h m$ and
    $\derclause \Gamma {\clause p u} m$}
  \\
  \infer
    {\fam i {\derclause {\Gamma_i} {\clause {p_i} {u_i}} m}}
    {\derclause
      {\envbigsum {\fam i {\Gamma_i}}}
      {\fam i {\clause {p_i} {u_i}}}
      m}

  \infer
  {\der {\Gamma, {\fam i {\of {x_i} {m_i}}}} u m}
  {\derclause \Gamma {\clause {\constr K {\fam i {x_i}}} u} m}
\\
  \fbox{Binding judgment $\derbinding {\fam i {\of {x_i} {\Gamma_i}}} b$}
\\
  \infer
  {\fam {i \in I}
    {\der
      {\Gamma_i, \fam {j \in I} {\of {x_j} {m_{i,j}}}}
      {t_i}
      \Return}
   \\
   \fam {i,j} {m_{i,j} \preceq \Guard}
   \\\\
   \fam i {\Gamma'_i =
     \envsum {\Gamma_i} {\envbigsum {\fam j {\mcomp {m_{i,j}} {\Gamma'_j}}}}}
  }
  {\derbinding
    {\fam {i \in I} {\of {x_i} {\Gamma'_i}}}
    {\fam {i \in I} {x_i = t_i}}}
\end{mathparfig}

\paragraph{Inference rules}
\myref{Figure}{fig:rules} presents the inference rules for access/usage modes.
The rules are composed into several different judgments, even though
our simple core language makes it possible to merge them. In the full
system for OCaml the decomposition is necessary to make the system
manageable.

\myfullref{Section}{subsec:examples} contains examples of mode
judgments in our system, corresponding to recursive definitions that
are accepted or rejected. Looking at those examples in parallel may
help understand some of the inference rules, in particular for
$\kwd{let}~\kwd{rec}$.

\paragraph{Variable and subsumption rules} The variable rule is as one
would expect: the usage mode of $x$ in an $m$-context is $m$. In this
declarative presentation, we let the rest of the environment $\Gamma$ be
arbitrary; we could also have imposed that it map all variables to
$\Ignore$. Our algorithmic check returns the ``least
demanding'' environment $\Gamma$ for all satisfiable judgments, so it
uses $\Ignore$ in any case.

We have a subsumption rule; for example, if we want to check
$t$ under the mode $\Guard$, it is always sound
to attempt
to check it under the stronger mode $\Dereference$.
Our algorithmic check will never use this rule; it is here for
completeness.
The direction of the comparison may seem unusual.
We can coerce
a $\der \Gamma t m$ into $\der \Gamma t m'$ when $m \succ m'$ holds,
while we might expect $m \leq m'$.
This comes from the fact that our backwards reading is opposite to
the usual reading direction of type judgments, and influenced our
order definition. When $m \succ m'$ holds, $m$ is \textit{more demanding}
than $m'$, which means (in the usual subtyping sense) that it
classifies \emph{fewer} terms.

\paragraph{Simple rules}

We have seen the $\lam x t$ rule already, in
\myref{Section}{subsec:intro-rules}. Since $\lambda$ delays evaluation,
checking $\lam x t$ in a usage context $m$ involves checking the body $t$ under
the weaker mode $\mcomp m \Delay$.
The necessary constraints $\Gamma$ are returned, after removing the
constraint over $x$%% the function variable.
\longfootnote{In situations where it
  is desirable to have a richer mode structure to analyze function
  applications, as considered by some of the related work
  (\myref{Section}{subsec:related-work}), we could use the mode $m_x$
  in a richer return mode $m_x \to m$.}.

The application rule checks both the function and its argument in
a $\Dereference$ context, and merges the two resulting environments,
taking the maximum (most demanding) mode on each side; a variable $y$
is dereferenced by $\app t u$ if it is dereferenced by either $t$ or
$u$.

The constructor rule is similar to the application rule, except that
the constructor parameters appear in $\Guard$ context, rather than
$\Dereference$.

\paragraph{Pattern-matching} The rule for $\match t h$ relies on
a different \emph{clause judgment} $\derclause \Gamma h m$ that checks each
clause in turn and merges their environments. On a single clause
$\clause {\constr K {\fam i {x_i}}} u$, we check the right-hand-side
expressions $u$ in the ambient mode $m$, and remove the pattern-bound
variables $\fam i {x_i}$ from the environment.\longfootnote{If we
  wanted a finer-grained analysis of usage of the sub-components of
  our data, we would use the sub-modes $\fam i {m_i}$ of the pattern
  variables to enrich the datatype of the pattern
  scrutinee.}

\paragraph{Recursive definitions} The rule for mutually-recursive
definitions $\letrecin b u$ is split into two parts with disjoint
responsibilities. First, the binding judgment
$\derbinding {\fam i {\of {x_i} {\Gamma_i}}} b$ computes, for each
definition $x_i = e_i$ in a recursive binding $b$, the usage
$\Gamma_i$ of the ambient context before the recursive binding --- we
detail its definition below.

Second, the $\letrecin b u$ rule of the term judgment takes these
$\Gamma_i$ and uses them under a composition
$\mcomp {m'_i} {\Gamma_i}$, to account for the actual usage mode of
the variables. (Here $\mcomp m \Gamma$ denotes the pointwise lifting of
composition for each mode in $\Gamma$.) The usage mode $m'_i$ is
a combination of the usage mode in the body $u$ and $\Guard$, used to
indicate that our call-by-value language will compute the values now,
even if they are not used in $u$, or only under a delay --- see
\myfullref{Remark}{rem:discarding}.

\paragraph{Deriving a simple $\kwd{let}$ rule} Before we delve into
the more general rule for mutually-recursive definitions, let us
mention the particular case of a single, non-recursive definition
$\letin x t u$. The general rule simplifies itself into the following:
\begin{mathline}
\infer
{\der {\Gamma_t} t \Return \\
 m' \defeq \max(m_x, \Guard) \\
 \der {\Gamma_u, \of x {m_x}} u m}
{\der
 {\envsum
   {\mcomp {m'} {\Gamma_t}}
   {\Gamma_u}}
 {\letin x t u}
 m}
\end{mathline}

\paragraph{Binding judgment and mutual recursion}
The \emph{binding judgment}
$\derbinding {\fam {i \in I} {\of {x_i} {\Gamma_i}}} b$ is independent
of the ambient context and usage mode; it checks recursive bindings
in isolation in the $\Return$ mode, and relates each name $x_i$
introduced by the binding $b$ to an environment $\Gamma_i$ on the
ambient free variables.

In the first premise, for each binding $(x_i = t_i)$ in $b$, we check
the term $t_i$ in a context split in two parts, some usage context
$\Gamma_i$ on the ambient context around the recursive definition, and
a context $\fam {j \in I} {\of {x_j} {m_{i,j}}}$ for the
recursively-bound variables, where $m_{i,j}$ is the mode of use of
$x_j$ in the definition of $x_i$.

The second premise checks that the modes $m_{i,j}$ are $\Guard$ or
less demanding, to ensure that these mutually-recursive definitions
are valid. This is the check mentioned at the end of
\myfullref{Section}{subsec:intro-right-to-left}.

The third premise makes mutual-recursion safe by turning the
$\Gamma_i$ into bigger contexts $\Gamma'_i$ taking transitive mutual
dependencies into account: if a recursive definition $x_i = e_i$ uses
the mutually-defined variable $x_j$ under the mode $m_{i,j}$, then we
ask that the final environment $\Gamma'_i$ for $e_i$ contains what you
need to use $e_j$ under the mode $m_{i,j} $, that is
$\mcomp {m_{i,j}} {\Gamma'_j}$. This set of recursive equations
corresponds to the fixed point of a monotone function, so in particular
it has a unique least solution.

Note that because the $m_{i,j}$ must be below $\Guard$, we can show
that $\mcomp {m_{i,j}} {\Gamma_j} \preceq \Gamma_j$. In particular, if
we have a single recursive binding, we have
$\Gamma_i \succeq \mcomp {m_{i,i}} {\Gamma_i}$, so the third premise
is equivalent to just $\Gamma'_i \defeq \Gamma_i$: the $\Gamma'_i$ and
$\Gamma_i$ only differ for non-trivial mutual recursion.

\paragraph{Unique minimal environment}

In \myappendixfullref{ann:properties} we develop some direct
meta-theoretic properties of our inference rules. We summarize here
the key results.
For each $t : m$, there exists a \emph{minimal} environment $\Gamma$ such that
$\der \Gamma t m$ holds.
\begin{theorem}[Principal environments]
  \label{thm:principal-environments}
  Whenever both $\der {\Gamma_1} t m$ and $\der {\Gamma_2} t m$ hold,
  then $\der {\min(\Gamma_1, \Gamma_2)} t m$ also holds.
\end{theorem}
We also define minimal \emph{derivations}, which restrict the
non-determinism in the variable and binding rules, and the way
subsumption may be used. Minimal derivations have minimal environments
and a syntax-directed structure. They precisely characterize the
behavior of our algorithm, implemented in the OCaml compiler: given
$t : m$, it constructs a minimal derivation of the form
$\der \Gamma t m$, and returns the environment $\Gamma$, which is
minimal for $t : m$.

\subsection{Examples}
\label{subsec:examples}

Our checker accepts a definition $\letrec {x = t}$ if there exists
a mode judgment $\der \Gamma t \Return$ that assigns a mode to $x$ in
$\Gamma$ that is $\Guard$ or smaller. The definition is rejected if
the mode of $x$ in the minimal judgment is $\Return$ or
$\Dereference$. Let us discuss various examples of definitions that are
accepted or rejected, along with the corresponding minimal judgments.

\paragraph{Separating $\Return$ from $\Guard, \Delay$}

The definitions $\letrec {x = \constr {\kwd{Fix}} x}$ or
$\letrec {f = \lam x {\app f x}}$ are valid: the definitions admit the
mode judgments $\der {\of x \Guard} {\constr {\kwd{Fix}} x} \Return$
and $\der {\of f \Delay} {\lam x {\app f x}} \Return$. On the other
hand, the definition $\letrec {x = x}$ is invalid, as the best
judgment for its body is $\der {\of x \Return} x \Return$, with
$\of x \Return$ (stricter than $\Guard$) in $\Gamma$.

\paragraph{Separating $\Guard$ from $\Delay$}

In the valid example $\letrec {f = \lam x {\app f x}}$, the subterm
$\app f x$ is dereferencing $f$; the mode of $f$ in the outer
environment is still $\Delay$ thanks to the composition
$\mcomp \Delay \Dereference = \Delay$. On the other hand, if we had
moved the application outside the delay,
$\letrec {f = \app {(\lam x f)} u}$, this definition would be
rejected, as we would have $\of f \Dereference$ thanks to the
composition $\mcomp \Dereference \Delay = \Dereference$.

$\Guard$ behaves differently than $\Delay$: our constructors are
strict, so dereferencing inside a guard is also a dereference. For
example, considering a function $g$ defined outside, both
$\letrec {x = \app g {(\constr {\kwd{Fix}} x})}$ and
$\letrec {x = \constr {\kwd{Fix}} {(\app g x)}}$ are rejected, as the
mode of $x$ is $\Dereference$ in the right-hand side of the
definition. In the first case this mode comes from the composition
$\mcomp \Dereference \Guard = \Dereference$, in the second case from
$\mcomp \Guard \Dereference = \Dereference$.

\paragraph{Separating $\Delay$ from $\Ignore$} Notice that if we have
$\der {\of x \Delay} {\plug t x} \Return$ then we have
$\der {\of x \Dereference, \of g \Dereference} {\app g {(\plug t x)}} \Return$
but if we have $\der {\of x \Ignore} {\plug t x} \Return$ then we have
$\der {\of x \Ignore, \of g \Dereference} {\app g {(\plug t x)}} \Return$.
So a declaration of the form $\letrec {x =\app g {(\plug t x})}$ is
rejected if $t$ uses the variable $x$ at mode $\Delay$, but accepted
if $x$ is not used ($\of x \Ignore$).

\paragraph{Separating $\Return$ from $\Dereference$} Similarly, if we
have $\der {\of x \Return} {\plug t x} \Return$ (for example
$\plug t x$ is $x$ or $(\letin y x y)$), then we have
$\der {\of x \Guard} {\constr {\kwd{Fix}} {(\plug t x)}} \Return$, but
if we have $\der {\of x \Dereference} {\plug t x} \Return$ then we have
$\of x \Dereference$ in the environment of
${\constr {\kwd{Fix}} {(\plug t x)}}$. In particular,
$\letrec {x = \constr {\kwd{Fix}} {(\plug t x})}$ is accepted in the
first case and rejected in the second.

\paragraph{Simple $\kwd{let}$ examples} (These examples are easier to
follow by using the simple $\kwd{let}$ rule than the general
$\kwd{let}~\kwd{rec}$ rule.). What is the mode $x$ in
$\letin z {\lam y {\constr {\kwd{Pair}} {(x, y)}}} {\constr {\kwd{Fix}} z}$?
The mode of $x$ in the definition
$\lam y {\constr {\kwd{Pair}} {(x, y)}}$ is $\Delay$, and the mode of
$z$ in the body $\constr {\kwd{Fix}} z$ is $\Guard$. The final context
(at global mode $\Return$) is $\mcomp \Guard {\of x \Delay}$, that is
$\of x \Delay$.

In the case of
$\letin z {\lam y {\constr {\kwd{Pair}} {(x, y)}}} {\app g z}$, the
final context is
$\mcomp \Dereference {\of x \Delay} = \of x \Dereference$. Finally,
the premise $m' = \max{(m, \Guard)}$ of the rule comes into play when
the body delays or ignores the defined variable: in the case of
$\letin z {\app g x} y$, the mode of $z$ in $y$ is $\Ignore$, but
the mode of $x$ in the whole term is not
$\mcomp \Ignore {\of x \Dereference}$, which would be $\of x \Ignore$,
but rather $\mcomp {\max(\Ignore,\Guard)} {\of x \Dereference}$, which
is $\of x \Dereference$.

\paragraph{$\kwd{let}~\kwd{rec}$ examples} The delicate aspect of the
$\derbinding {\fam i {\of {x_i} {\Gamma_i}}} {\fam i {x_i = t_i}}$
judgment is the fixpoint of equations
$ \Gamma'_i = \envsum {\Gamma_i} {\envbigsum {\fam j {\mcomp {m_{i,j}} {\Gamma'_j}}}} $,
which computes a ``transitive closure'' of usage modes of the
mutually-recursively-defined variables. Consider for example the term $t$ defined as
\begin{mathline}
  t
  \quad
  \defeq
  \quad
  \letrecin {
    x' = x
    \andlet
    y = \constr {\kwd{Fix}} {(x')}
    \andlet
    z = \app g y
  } z
\end{mathline}
We index the three definitions by $i \in I$ with
$I \defeq \{x', y, z\}$. The contexts $\fam i {\Gamma_i}$ of the rule
correspond to the dependency of each right-hand-side on
non-mutually-recursive variables. The $\fam i {\Gamma'_i}$ are defined
by a system of recursive equations, depending on each $\Gamma_j$ and
the mode of use of the variable $j$ in the definition of $i$. We have:
\begin{mathline}
  \begin{array}{l@{\defeq}l}
    \Gamma_{x'} & (\of x \Return) \\
    \Gamma_y & \emptyset \\
    \Gamma_z & (\of g \Dereference)
  \end{array}

  \begin{array}{l@{\defeq}l}
    \Gamma'_{x'} & \Gamma_x \\
    \Gamma'_y & \envsum {\Gamma_y} {\mcomp \Guard {\Gamma'_{x'}}}  \\
    \Gamma'_z & \envsum {\Gamma'_z} {\mcomp \Dereference {\Gamma'_y}} \\
  \end{array}
\end{mathline}
The smallest fixpoint solution has $\Gamma'_{x'} = (\of x \Return)$,
$\Gamma'_y = (\of x \Guard)$, and
$\Gamma'_z = (\of g \Dereference, \of x \Dereference)$. In particular,
notice how $x$ is accessed at mode $\Dereference$ by $z$, even though
it does not syntactically appear in its definition. The whole term
$t$, in the ambient mode $\Return$, is in the environment
$\Gamma'_z = (\of g \Dereference, \of x \Dereference)$. If we had used
a simpler $\kwd{rec}$ rule that would return the $\fam i {\Gamma_i}$
instead of the $\fam i {\Gamma'_i}$, immediate usage rather than
transitive usage, $t$ would typed in the environment
$\Gamma_z = \of g \Dereference$, that is with $\of x \Ignore$. This
would be unsound, for example the vicious definition
$\letrec {x = t}$ would be accepted.

\begin{version}{\Long}
\subsection{Discussion}

\paragraph{Declarative vs. algorithmic rules} A presentation of a type
system can be more ``algorithmic'' or more ``declarative'', sitting on
a continuous spectrum. More-declarative systems have convenient typing
rules corresponding to reasoning principles that are sound in the
metatheory, but may be harder to implement in practice. More-algorithmic
systems have more rigid inference rules, that are easier
to implement as a checking or inference system (typically they may be
syntax-directed); some valid reasoning principles may not be available
as rules but only as admissible rules (requiring a global rewrite of
the derivation) or not at all.

In our system, the variable and subsumption rules are typically
declarative: the variable rule has an undetermined $\Gamma$ context, and
the subsumption rule makes the system non-syntax-directed. Without
those two rules, it would not be possible to prove
$\der {x : \Guard, y : \Dereference} y \Return$, but only the
stronger judgment $\der {x : \Ignore, y : \Return} y \Return$.

\sloppy
On the other hand, our binding rule has an algorithmic flavor: it
introduces a family of contexts $\fam i {\Gamma'_i}$ that is uniquely
determined as the solution of a system of recursive equations
$
\fam i {\Gamma'_i =
  \envsum {\Gamma_i} {\envbigsum {\fam j {\mcomp {m_{i,j}} {\Gamma'_j}}}}}
$,
so its application requires computing a fixpoint.  A more declarative
presentation would allow ``guessing'' any family $\fam i {\Gamma_i}$
that satisfies the inequations necessary for soundness:
\begin{mathline}
\infer
  {\fam {i \in I}
    {\der
      {\Gamma_i, \fam {j \in I} {\of {x_j} {m_{i,j}}}}
      {t_i}
      \Return}
   \\
   \fam {i,j} {m_{i,j} \preceq \Guard}
   \\\\
   \fam i {\Gamma_i \succeq
     \envsum {\Gamma_i} {\envbigsum {\fam j {\mcomp {m_{i,j}} {\Gamma_j}}}}}
  }
  {\derbinding
    {\fam {i \in I} {\of {x_i} {\Gamma_i}}}
    {\fam {i \in I} {x_i = t_i}}}
\end{mathline}
\end{version}

\fussy

\paragraph{Backwards type systems}

Typing rules are a specialized declarative language to describe and
justify various computational processes related to a type system
(type checking, type inference, elaboration, etc.). Our mode system
read ``backwards'' is one possible way to describe the static analysis we
are capturing, which could also be described in many other ways: as
pseudocode, as a fixpoint of equations, through a denotational
semantics, etc. In general we believe that reading type systems
backwards can give a nice, compact, declarative presentation of
certain demand analyses, in a language that type designers are already
familiar with.

Our terminology follows the ``backward analysis'' notion described for
logic programming languages by
\citet*{genaim-codish}, i.e.~our algorithm answers the question (posed in that work) ``Given
a program and an assertion at a given program point, what are the
weakest requirements on the inputs to the program which guarantee that
the assertion will hold whenever execution reaches that point?''
For our analysis, the \emph{assertion} is the mode under which an
expression is checked, and the \emph{requirements on the inputs}
correspond to the computed environment $\Gamma$.

Backward analyses for functional languages also appear in work by
Hughes~\citep{backward-analysis}.
A notable difference is that Hughes-style ``demand analyses''\footnote{
For a recent example, see the unpublished
draft \citet*{haskell-demand-analysis}. Thanks are due to Joachim
Breitner for the reference.} are typically presented in a denotational
style, using tools of domain theory.
However, some recent work (such as the cardinality analysis
by \citet{DBLP:journals/jfp/SergeyVJB17}) presents backward analyses
in a syntactic style more similar to that used in the present paper.

%% We think that some of these works could be presented
%% in our style as well. This could make them simpler to approach and
%% work with, or at least accessible to a different audience. For
%% example, our main soundness proof in
%% \myfullref{Section}{section:meta-theory} is done simply as
%% a subject-reduction proof, which is familiar to many type-system
%% designers.

\paragraph{Modes as modalities} Untyped or dynamically-typed languages
can be seen as ``uni-typed'', with a ``universal type'' $\Dyn$ of all
values. Language constructions can be presented as section/retraction
pairs from $\Dyn$ to a type that computes, such as $\Dyn \to \Dyn$ for
functions or $\Dyn \times \Dyn \times \dots \times \Dyn$ for tuples;
for example, the untyped term $\app {(\lam x t)} u$ can be
explicitated into $\mathsf{app}~(\mathsf{lam}(\lam x t))~u$, for
combinators $\mathsf{app} : \Dyn \to (\Dyn \to\Dyn )$ and
$\mathsf{lam} : (\Dyn \to \Dyn) \to \Dyn$ such that
$\mathsf{app} \circ \mathsf{lam}$ is the identity on
$\Dyn \to \Dyn$ --- but $\mathsf{lam} \circ \mathsf{app}$ gets stuck on
non-functions.

Rather than \emph{types}, it is more precise to see our modes as
\emph{modalities} on this universal type $\Dyn$. The hypothesis $x : m$ in
a context would be a modal hypothesis $x :^m \Dyn$, and the section
combinators are given modal types. Within the modal framework of
\citet*{unified-view-of-modalities} for example, writing
$\modfun m A B$ for the modal function type, some of our typing rules
could be modelled with
$\mathsf{lam}: \modfun \Delay {\modfun \Dereference \Dyn \Dyn} \Dyn$,
$\mathsf{app}: \modfun \Dereference \Dyn {\modfun \Dereference \Dyn \Dyn}$, and
for datatypes something like
$\mathsf{pack}^K_d : \modfun \Guard \Dyn^d \Dyn$, and
$\mathsf{unpack}^K_d : \modfun \Dereference \Dyn {(\Dyn + \Dyn^d)}$, where $K$
is a constructor of arity $d$, the $\Dyn +$ return value of unpacking
indicates an incompatible constructor, and $\Dyn^d$ is
$\Dyn \times \Dyn \times \dots$, $d$ times.

This view naturally extends to supporting finer-grained analyses and
abstraction of the sort found in other work (e.g.~the system
introduced by~\citet{dreyer-2004}).  In our system every function
argument has mode $\Dereference$; a refinement that
allowed types and modes to interact could support a range of modes for
functions that used their arguments in different ways.
However, this additional expressivity would require substantial
changes to OCaml's type system, and a proper treatment of abstraction
would require some form of mode polymorphism.  (For example, in the
appliction function \code{let h g x = g x}, mode polymorphism is
required to support applying the \code{h} to all possible
functions \code{g}, some of which dereference their argument, and some
of which do not.)

Our aim of replacing OCaml's existing syntactic check with a more
principled version way did not justify this substantial additional
complexity.
More generally, our view is that abstraction over modes is
more suited to module systems (for which the system described by
\citet{dreyer-2004} was developed), where types are already
explicit, than to the term languages that our system is designed for.
This view is supported by the fact that, as far as we know, although
existing languages support a variety of checks on well-formedness of
recursive definitions, there has been (as far as we know) no attempt
in any of these to incorporate a system with support for mode
abstraction.

\section{Meta-Theory: Soundness}
\label{section:meta-theory}

\subsection{Operational Semantics}
\label{subsec:operational-semantics}

\begin{mathparfig}[t!]{fig:semantics}{Operational semantics}
  \begin{array}{r@{~}r@{~}l}
    \set{Values} \ni v & \bnfeq
    & \lam x t
      \bnfor \constr K {\fam i {w_i}}
      \bnfor \plug L v
    \\
    \set{WeakValues} \ni w & \bnfeq
    & x,y,z \bnfor v \bnfor \plug L w \\
    \\
    \set{ValueBindings} \ni B & \bnfeq
    & \fam i {x_i = v_i} \\
    \set{BindingCtx} \ni L
    & \bnfeq
    & \square \bnfor \letrecin B L \\
  \end{array}
  \begin{array}{r@{~}r@{~}l}
    \set{EvalCtx} \ni E
    & \bnfeq
    & \square \bnfor \plug E F \\
    \set{EvalFrame} \ni F
    & \bnfeq & \app \hole t
    \;\bnfor\; \app t \hole \\
    & \bnfor & \constr K {(\fam i {t_i}, \hole, \fam j {t_j})} \\
    & \bnfor & \match \hole h \\
    & \bnfor & \letrecin {b, x = \hole, b'} u \\
    & \bnfor & \letrecin B \hole \\
  \end{array}
  \\
  \infer{ }
  {\rewhead {\app {\plug L {\lam x t}} v} {\plug L {\subst t {\by v x}}}}

  \infer
  {\forall {(\clause {\constr {K'} {\fam j {x'_j}}} u')} \in h,\; K \neq K'}
  {\rewhead
    {\match {\plug L {\constr K {\fam i {w_i}}}}
      {(h \mid \clause {\constr K {\fam i {x_i}}} u \mid h')}}
    {\plug L {\substfam u i {\by {w_i} {x_i}}}}}
  \\
  \infer
  {\rewhead t {t'}}
  {\rew {\plug E t} {\plug E {t'}}}

  \infer
  {\inctx x v E}
  {\rew {\plug E x} {\plug E v}}

  \infer
  {\inframe x v F \;\vee\; \inctx x v E}
  {\inctx x v {\plug E F}}

  \infer{\inbinding x v B}
  {\inframe x v {\letrecin B \hole}}

  \infer{\inbinding x v (b \cup b')}
  {\inframe x v {\letrecin {b, y = \hole, b'} u}}
\end{mathparfig}

\myref{Figure}{fig:semantics} presents our operational semantics,
largely reused from \citet*{dynamic-semantics-2008} with extensions
(support for algebraic datatypes) and changes (use of \emph{reduction
  at a distance}). Unless explicitly noted, the content and ideas in
this \myref{Subsec}{subsec:operational-semantics} come from their
work.

\paragraph{Weak values} As we have seen, constructors in recursive
definitions can be used to construct cyclic values. For example, the
definition
$\letrec {x = \constr {\kwd{Cons}} {(\constr {\kwd{One}} {(\emptyset)}, x)}}$
is normal for this reduction semantics. The occurrence of the variable
$x$ inside the $\kwd{Cons}$ cell corresponds to a back-reference, the
cell address in a cyclic in-memory representation.

This key property is achieved by defining a class of \emph{weak
  values}, noted $w$, to be either (strict) values or variables. Weak
values occur in the definition of the semantics wherever a cyclic
reference can be passed without having to dereference.

Several previous works (see \myfullref{Section}{subsec:related-work})
defined semantics where $\beta$-redexes have the form
$\app {(\lam x t)} w$, to allow yet-unevaluated recursive definitions
to be passed as function arguments. OCaml does not allow this
(a function call requires a fully-evaluated argument),
\Xgabriel{Is this actually true? Thinking more about this, I suspect
  that OCaml would be fine if passed partially-evaluated values
  (I thought of issues with the GC exploring them, but in fact normal
  computation of recursive values already exposes
  partially-uninitialized values to the GC.) I don't suggest relaxing
  the semantics here, but this sentence should probably be removed.}
so our redexes are the traditional $\app {(\lam x t)} v$. This is
a difference from \citet*{dynamic-semantics-2008}. On the other hand,
we do allow cyclic datatype values by only requiring weak values under
data constructors: the corresponding value form is
$\constr K {\fam i {w_i}}$.

\paragraph{Bindings in evaluation contexts} An \emph{evaluation
  context} $E$ is a stack of \emph{evaluation frames} $F$ under
which evaluation may occur. Our semantics is under-constrained
(for example, $\app t u$ may perform reductions on either
$t$ or $u$), as OCaml has unspecified evaluation order for
applications and constructors, but making it deterministic would not
change much.

\sloppy
One common aspect of most operational semantics for \code{let rec},
ours included, is that $\letrecin B \hole$ can be part of evaluation
contexts, where $B$ represents a recursive ``value binding'', an
island of recursive definitions that have all been reduced to
values. This is different from traditional source-level operational
semantics of $\letin x v u$, which is reduced to $\subst u {\by v x}$
before going further. In $\kwd{letrec}$ blocks this substitution reduction
is not valid, since the value $v$ may refer to the name $x$, and so
instead ``value bindings'' remain in the context, in the style of
explicit substitution calculi. We call these context fragments
``binding contexts'' $L$.

\fussy

\paragraph{Head reduction} Head redexes, the sources of the
head-reduction relation $\rewhead t {t'}$, come from applying
a $\lambda$-abstraction or from pattern-matching on a head
constructor. Following ML semantics, pattern-matching is
ordered: only the first matching clause is taken.

One mildly original feature of our head reduction is the use of
\emph{reduction at a distance}, where binding contexts $L$ are allowed
to be presented in the middle of redexes, and lifted out of the
reduced term. This presentation is common in explicit-substitution
calculi\footnote{See for example \citet*{accattoli:hal-00528228},
  which links to earlier references on the technique.}, as it gives
the minimal amount of lifting of explicit substitutions required to
avoid blocking reduction. In the calculus of
\citet*{dynamic-semantics-2008}, lifting was permitted in arbitrary
positions by the Merge rule. For example, the reduction sequence
$\rewstar {\app {(\letrecin B {\lam x t})} {v}} {\letrecin B {\subst t {\by u x}}}$
is admissible in both systems, but the ``useless'' reduction
$\rewstar {\app {(\letrecin B x)} {v}} {\letrecin B {\app x v}}$ is
not present in our system. Reduction at a distance tends to make
definitions crisper and simplify proofs.\footnote{For an example of
  beneficial use of reduction-at-distance in previous work from the
  rewriting community, see the at-a-distance presentation of the
  $\pi$-calculus in \citet*{pi-calculus-at-a-distance}.}

\paragraph{Reduction} Reduction $\rew t {t'}$ may happen under any
evaluation context. The first reduction rule is %% completely
standard: any redex $\plug H v$ can be reduced under an evaluation
context $E$.

The second rule reduces a variable $x$ in an evaluation context $E$
by binding lookup: it is replaced by the value of the recursive
binding $B$ in the context $E$ which defines it. This uses the
auxiliary definition $\inctx x v E$ to perform this lookup.

The lookup rule has worrying consequences for our rewriting relation:
it makes it nondeterministic and non-terminating. Indeed, consider a
weak value of the form $\constr K {(x)}$ used, for example, in a
pattern-matching $\match {\constr K {(x)}} h$. It is possible to
reduce the pattern-matching immediately, or to first lookup the value
of $x$ and then reduce. Furthermore, it could be the case that $x$ is
precisely defined by a cyclic binding $x = \constr K {(x)}$. Then the
lookup rule would reduce to $\match {\constr K {(\constr K (x))}} h$,
and we could keep looking
indefinitely. \citet*{dynamic-semantics-2008} discuss this in detail
and prove that the reduction is in fact confluent modulo
unfolding. (Allowing these irritating but innocuous behaviors is a
large part of what makes their semantics simpler than previous
presentations.)

\paragraph{Example}
\newcommand{\cS}[1]{\constr {\kwd{S}} #1}

Consider the following program:
\begin{mathpar}
    \begin{array}{l}
      \letrecin {\infty = {(\letrecin {x = \cS x} x})} {}
      \\
      \quad \match \infty
               {(\kwd{Z} \to \kwd{None} \mid {\cS y} \to \constr {\kwd{Some}} y)}
    \end{array}
\end{mathpar}
The first binding $\letrecin {x = \cS x} x$ is not a value yet, only
a weak value. The first reduction this program can take is to lookup
the right-hand-side occurence of $x$:
\begin{mathpar}
  \begin{array}{l}
    \begin{array}{l}
       \rew{}
        {\letrecin {\infty = {(\letrecin {x = \cS x} {\cS x})}} {}}
        \\
        \qquad \match \infty
               {(\kwd{Z} \to \kwd{None} \mid {\cS y} \to \constr {\kwd{Some}} y)}
      \end{array}
  \end{array}
\end{mathpar}
After this reduction $\infty$, is bound to a value, so it can in turn
be looked up in $\kwd{match}~\infty$:
\begin{mathpar}
  \begin{array}{l}
    \begin{array}{l}
      \rew{}
      {\letrecin {\infty = {(\letrecin {x = \cS x} {\cS x})}} {}}
      \\
      \qquad \match {(\letrecin {x = \cS x} {\cS x})} {}
      \\
      \qquad {(\kwd{Z} \to \kwd{None} \mid {\cS y} \to \constr {\kwd{Some}} y)}
    \end{array}
  \end{array}
\end{mathpar}
At this point we have a $\kwd{match}$ redex of the form $\kwd{match}~\plug L {\cS x}$,
which gets reduced by lifting the binding context $L$:
\begin{mathpar}
  \begin{array}{l}
    \begin{array}{l}
      \rew{}
      {\letrecin {\infty = {(\letrecin {x = \cS x} {\cS x})}} {}}
      \\
      \qquad
      \letrecin {x = \cS x} {\constr {\kwd{Some}} x}
    \end{array}
  \end{array}
\end{mathpar}

\subsection{Failures}
\label{subsec:failures}

In this section, we are interested in formally defining dynamic
failures. When can we say that a term is ``wrong''? --- in particular,
when is a valid implementation of the operational semantics allowed to
crash? This aspect is not discussed in detail by
\citet*{dynamic-semantics-2008}, so we had to make our own
definitions; we found it surprisingly subtle.

The first obvious sort of failure is a type mismatch between a value
constructor and a value destructor: application of
a non-function, pattern-matching on a function instead of a head
constructor, or not having a given head constructor covered by the
match clauses. These failures would be ruled out by a simple type
system and exhaustivity check.

The more challenging task is defining failures that occur when trying to
access a recursively-defined variable too early. The lookup
reduction rule for a term $\plug E x$ looks for the value of $x$ in
a binding of the context $E$. This value may not exist (yet), and that
may or may not represent a runtime failure.

We assume that bound names are all distinct, so there
may not be several $v$ values. The only binders that we reduce under
are \code{let rec}, so $x$ must come from one; however, it is
possible that $x$ is part of a \code{let rec} block currently being
evaluated, with an evaluation context of the form
$E[\letrecin {(x = t, y = E')} u]$ for example, and that $x$'s binding has not yet
been reduced to a value.

However, in the presence of data constructors that permit building cyclic
values not all such cases are failures. For example the term
$\letrecin {x = \constr {\kwd{Pair}} {(x, t)}} x$ can be decomposed
into $\plug E x$ to isolate the occurrence of $x$ as the first member
of the pair. This occurrence of $x$ is in reducible position, but
there is no $v$ such that $\inctx x v E$, unless $t$ is already a weak
value.

To characterize failures during recursive evaluation, we propose to
restrict ourselves to \emph{forcing contexts}, denoted $\forcing E$,
that must access or return the value of their hole. A variable in a forcing
context that cannot be looked up in the context is a dynamic failure:
we are forcing the value of a variable that has not yet been
evaluated. If a term contains such a variable in lookup position, we
call it a \emph{vicious} term.
\myref{Figure}{fig:failures} gives a precise definition of these failure terms.

\begin{mathparfig}[h!]{fig:failures}{Failure terms}
  \begin{array}{r@{~}r@{~}l}
    \set{HeadFrame} \ni H
    & \bnfeq
    & \app \hole v
      \;\bnfor\; \match \hole h \\
    \set{ForcingFrame} \ni {\forcing F}
    & \bnfeq & \app \hole v
    \;\bnfor\; \app v \hole
    \;\bnfor\; \match \hole h \\
    \set{ForcingCtx} \ni {\forcing E}
    & \bnfeq & L
    \;\bnfor\; \plug E {\plug {\forcing F} L} \\
  \end{array}

  \set{Mismatch} \defeq \{ \plug E {\plug H v} \mid \nrewhead {\plug H v} \}

  \set{Vicious} \defeq \{ \plug {\forcing E} x
                  \mid \nexists v, \inctx x v {\forcing E} \}
\end{mathparfig}

Mismatches are characterized by \emph{head frames}, context fragments
that would form a $\beta$-redex if filled with a value of the
correct type. A term of the form $\plug H v$ that is stuck for
head-reduction is a constructor-destructor mismatch.

The definition of forcing contexts $\forcing E$ takes into account the
fact that recursive value bindings remain, floating around, in the
evaluation context. A forcing frame $\forcing F$ is a context fragment
that forces evaluation of its variable; it would be tempting to say
that a forcing context is necessarily of the form
$\hole$ or $\plug E {\forcing F}$, but for example
$\plug {\forcing F} {\letrecin B \hole}$ must also be considered
a forcing context.

Note that, due to the flexibility we gave to the evaluation order,
mismatches and vicious terms need not be stuck: they may have other
reducible positions in their evaluation context. In fact, a vicious
term failing on a variable $x$ may reduce to a non-vicious term if
the binding of $x$ is reduced to a value.

\subsection{Soundness}
\label{subsec:soundness}

The proofs of these results are in \myappendixref{ann:proofs}.

\begin{lemma}[Forcing modes]
  \label{lem:forcing-modes}
  If\;\;$\der {\Gamma, \of x {m_x}} {\plug {\forcing E} x} m$ with $m \succeq \Return$,
  then also $m_x \succeq \Return$.
\end{lemma}

\begin{theorem}[Vicious]\label{lem:vicious}
  $\der \emptyset t \Return$ never holds for $t \in \set{Vicious}$.
\end{theorem}

\begin{theorem}[Subject reduction]\label{thm:subject-reduction}
  If\;\;$\der \Gamma t m$ and $\rew t {t'}$ then $\der \Gamma {t'} m$.
\end{theorem}

\begin{corollary}\label{cor:typed-programs-cannot-go-vicious}
  $\Return$-typed programs cannot go vicious.
\end{corollary}

\begin{version}{\False} % [Gabriel] for later
\subsection{Completeness}

From the notes:
\begin{verbatim}
We could show that, for each potential relaxation of one of the typing
rules, we can build an example of invalid recursion.

If we did follow this approach exhaustively for all rules (some of
them only in appendix, of course), then the set of counter-examples
would form a formal proof of completeness of our rules. This is nice,
but not a hard goal --- the main point is to use the counter-examples
for reader intuition.
\end{verbatim}

One thing to be careful about is the amount of flexibility in the
evaluation order. Our check is agnostic to evaluation order, so it
rejects any definition that could fail at runtime in \emph{some}
evaluation order. If the reference semantics is ordered, there is
a slight mismatch here. We could either say that our check is complete
for any reordering of the definitions, or extend the
reference semantics to make evaluation non-deterministic
(stepping away from the reference citation a bit). \Xgabriel{I would
  prefer the second option, I think it gives the cleanest presentation
  overall, but we must be careful to explain the change.}
\end{version}

\section{Global Store Semantics}
\label{section:global-store}

In Section~\ref{subsec:operational-semantics}, we developed an
operational semantics for \texttt{letrec} using explicit substitutions
as ``local stores'' for recursive values. This technique was first
studied in more theoretical research communities (pure lambda-calculus
and rewriting), and imported in more applied programming-language
works in the 1990s~\citep*{felleisen-local-store} and in particular
local-store presentations of call-by-need and
\texttt{letrec}~\citep*{ariola-felleisen}.

Before local-store semantics were proposed, recursive values (and lazy
thunks) were modelled using a ``global store'' semantics, more closely
modelling the dummy-initialization-then-backpatching approach used in
real-world implementations.

Explicit-substitution, local-store semantics enjoy at least the
following advantages over the global-store semantics of
\texttt{letrec}.
\begin{itemize}
\item Our explicit-substitution semantics is defined directly at the
  level of our term syntax, without going through additional features
  (global store, initialization and setting of memory locations) that
  are not part of the source language. Source-level semantics are
  typically more high-level and more declarative than lower-level
  semantics; they allow to reason on the source directly, without
  going through an encoding layer.

  We can compare \texttt{letrec} to the problem of defining ``tail
  calls''. One can teach tail-calls and reason about them through
  a compilation to a machine with an explicit call stack. But there is
  a source-level explanation of tail-calls, by reasoning on the size
  (and the growth) of the evaluation context; this explanation is
  simpler and makes tail-calls easier to reason about.

\item Our local-store semantics allows more local reasoning: reducing
  a subterm in evaluation position only affects this subterm
  (which may contain explicit substitutions), instead of also
  affecting a global store. This makes it easier to define
  program-equivalence relations that are congruence (are stable
  under contexts), and generally to prove program-equivalence results.

  \begin{example}[Store duplication]\label{ex:store-duplication}
  For example, consider the two following programs, manipulating
  infinite lists of ones:
  \newcommand{\ones}{(\letrecin {x = \kwd{Cons}(1,x)} x)}
  \begin{mathpar}
    \letin o \ones {\app {\app f o} o}

    {\app {\app f \ones} \ones}
  \end{mathpar}
  With our explicit-substitution semantics, it is trivial to show that
  these two terms are equivalent: they reduce to the same term, namely
  \newcommand{\onesreduced}{(\letrecin {x = \kwd{Cons}(1,x)} {\kwd{Cons}(1,x)})}
  \begin{mathpar}
    {\app {\app f \onesreduced} \onesreduced}
  \end{mathpar}
  (Note that $\ones$ is not a value in our semantics, only a weak
  value; one needs to lookup $x$ once to get a value so that the
  $o$-binding can be reduced.)

  In a global-store semantics, on the contrary, it is not at all
  obvious that the two programs are equivalent; indeed, they reduce
  to configurations of the following forms:
  \begin{mathpar}
    ([x \mapsto \kwd{Cons}(1,x)], \app {\app f x} x)

    ([x_1 \mapsto \kwd{Cons}(1,x_1),
      x_2 \mapsto \kwd{Cons}(1,x_2)], \app {\app f {x_1}} {x_2})
  \end{mathpar}
  These configurations are not equivalent by reduction; 
  the equational theory would need a stronger equivalence principle
  that could de-duplicate bisimilar fragments of the store.
  \end{example}
\end{itemize}

However, since a large part of our community is unfamiliar with
local-store semantics, in the interest of accessibility, we also propose in
\Long{this section}{an extended version of this work} a
global-store semantics for our language. We show that the soundness
result for our analysis can be lifted to this semantics:
$\Return$-typed programs cannot go vicious in the global-store
semantics either. This is done by formalizing a compilation pass from
the local-store language to the
global-store+backpatching language, and showing a backward simulation
result --- effectively redoing the compilation-correctness work of
\citet*{hirschowitz-compilation-2003,hirschowitz-compilation-2009} in
our setting.

\begin{version}{\Long}
\subsection{Target Language}
\label{subsec:target-language}

The global-store language is called the ``target'' language as we will
prove the correctness of a compilation pass from the ``source''
(local-store) to ``target'' (global-store) languages. Its grammar and
operational semantics are given in
Figure~\ref{fig:target-syntax}.

A store location is created uninitialized with the $\newin x t$
constructor, and can be defined exactly once by the assignment
expression $\setref x t$; the heap binds locations to ``heap blocks''
that are either uninitialized ($\bot$) or a value. The write-once
discipline is enforced by the $\setref x t$ reduction rule, which
requires the heap block to be $\bot$. Trying to read an uninitialized
memory location is the dynamic error that our analysis prevents: it is
the target equivalent of vicious terms, and we call it a ``segfault''
by analogy with the unpleasant consequences of the corresponding error
in most compiled languages.

The reduction uses a distinguished data constructor, $\kwd{Done}$,
playing the role of a unit value in our untyped semantics: it is the
value returned by the evaluation of an assignment expression
$\setref x v$.

\begin{mathparfig}[t!]{fig:target-syntax}%
  {Syntax and reduction of the global-store (target) language}
  \begin{array}{l@{~}r@{~}l}
    \set{TTerms} \ni t, u & \bnfeq & x, y, z \\
    & \bnfor & \lam x t %  \\
    % &
      \bnfor % &
               \app t u \\
    & \bnfor & \constr K {\fam i {t_i}} % \\
    % &
      \bnfor % &
               \match t h \\
    & \bnfor & \newin x t \\
    & \bnfor & \setref x t
  \end{array}

  \begin{array}{l}
    \text{Handlers, patterns:} \\
    \text{as in the source language}
  \end{array}

  \begin{array}{l@{~}r@{~}l}
    \set{TValues} \ni v & \bnfeq
    & \lam x t \bnfor \constr K {\fam i {w_i}} \\
    \set{TWeakValues} \ni w & \bnfeq
    & x,y,z \bnfor v \\
  \end{array}

  \begin{array}{r@{~}r@{~}l}
    \set{ValueHeaps} \ni B & \bnfeq & \emptyset \bnfor \heapext B x v \\
    \set{Heaps} \ni H & \bnfeq & \emptyset \bnfor \heapext H x {v?} \\
    \set{HeapBlocks} \ni v? & \bnfeq & v \bnfor \bot \\
  \end{array}

  \begin{array}{l@{~}r@{~}l}
    \set{TEvalCtx} \ni E
    & \bnfeq
    & \square \bnfor \plug E F \\
    \set{TEvalFrame} \ni F
    & \bnfeq & \app \hole t
    \;\bnfor\; \app t \hole \\
    & \bnfor & \constr K {(\fam i {t_i}, \hole, \fam j {t_j})} \\
    & \bnfor & \match \hole h \\
    & \bnfor & \setref x \hole \\
  \end{array}

  \begin{array}{l@{~}r@{~}l}
    \set{TFEvalCtx} \ni {E_f}
    & \bnfeq
    & \square \bnfor \plug E {F_f} \\
    \set{TForcingFrame} \ni F_f
    & \bnfeq & \app \hole t
    \;\bnfor\; \app t \hole \\
    & \bnfor & \match \hole h \\
  \end{array}
\\
  \infer[ctx]
  {\rewhead {\conf {H_\square} {t_h}} {\conf {H'_\square} {t'_h}}}
  {\rew
    {\conf {\heapjoin {H_E} {H_\square}} {\plug E {t_h}}}
    {\conf {\heapjoin {H_E} {H'_\square}} {\plug E {t'_h}}}}

  \infer{ }
  {\rewhead
    {\conf \emptyset {\app {(\lam x t)} u}}
    {\conf \emptyset {\subst t {\by u x}}}}

  \infer
  {\forall {(\clause {\constr {K'} {\fam j {x'_j}}} u')} \in h,\; K \neq K'}
  {\rewhead
    {\conf \emptyset
      {\match {\constr K {\fam i {w_i}}}
        {(h \mid \clause {\constr K {\fam i {x_i}}} u \mid h')}}}
    {\conf \emptyset
      {\substfam u i {\by {w_i} {x_i}}}}}

  \infer[new]
  { }
  {\rewhead
    {\conf \emptyset {\newin x t}}
    {\conf {\heapext {} x \bot} t}}

  \infer[set]
  { }
  {\rewhead
    {\conf {\heapext {} x \bot} {\setref x v}}
    {\conf {\heapext {} x v} \Done}
  }

  \infer[lookup]
  { }
  {\rewhead
    {\conf {\heapext {} x v} x}
    {\conf {\heapext {} x v} v}}

  \mathsf{Segfault} \defeq \conf {\heapext H x \bot} {\plug {E_f} x}
\end{mathparfig}

\begin{remark}[Variables as locations]
  We use variables as heap locations, an idea that goes back at least
to \cite{DBLP:conf/popl/Launchbury93}. This is equivalent to having
a distinguished syntactic category of locations thanks to the
Barendregt convention --- we work modulo $\alpha$-equivalence and can
thus always assume that bound variables are distinct from free
variables of the same term. Picking a ``fresh enough'' variable $x$ each
time we encounter a $\kwd{new}$ binder gives a new memory location. In
particular, those variables do not correspond to static binding
positions in the program we started computing; each time a function is
called, the $\kwd{new}$ binders in its body bind over
$\alpha$-equivalent fresh names. This technical choice will make it
easier to relate to the source language, where recursive values are
bound to variables in scope.
\end{remark}

\subsection{Parallel or Order-Independent Evaluation}

In Section~\ref{subsec:target-language} we chose a distinguished
$\kwd{Done}$ constructor to represent a unit value in our target
language. We now choose another distinguished constructor
$\constr {\kwd{Par}} {(t_1, \dots, t_n)}$ to represent a product of
values to evaluate in an arbitrary order. We use a single constructor
at arbitrary arities, but we could just as well use a family of
constructors $\kwd{Par}_n$.

\begin{definition}[Sequential and parallel evaluation]~\\
We define the following syntactic sugar:
\begin{mathpar}
  \begin{array}{l@{\quad}l@{\quad}l}
    (\seq t u)
    & \defeq &
    (\match t (\clause \Done u))
  \\
    \para \fam {i \in I} {t_i}
    & \defeq &
    (\match
      {\constr {\kwd{Par}} {\fam {i \in I} {t_i}}}
      {\clause {\constr {\kwd{Par}} {\fam {i \in I} {x_i}}} \Done})
  \end{array}
\end{mathpar}
\end{definition}

The term $\para {\fam i {t_i}}$ represents the evaluation of a family
of $\Done$-returning term, in an arbitrary order. In particular,
$(\seq \hole u)$ and
$\para {(\fam {i \in I} {t_i}, \hole, \fam {j \in J} {t_j})}$ are
evaluation contexts, and the following reduction rules are derivable:
\begin{mathpar}
  \rewhead {\conf \emptyset {(\seq \Done u)}} {\conf \emptyset u}

  \rewhead {\conf \emptyset {{\para {\fam {i \in I} \Done}}}} {\conf \emptyset \Done}
\end{mathpar}

\subsection{Translation: Compiling \texttt{letrec} into Backpatching}

We define in Figure~\ref{fig:compilation-into-store} our translation $\comp t$
from source to target terms, explaining recursive value bindings in
terms of backpatching. The interesting case is $\comp {\letrecin b u}$, the others
are just a direct mapping on all subterms.
\begin{mathpar}
  \comp {\letrecin {\fam i {x_i = t_i}} u}

  \defeq

  \newin {\fam i {x_i}}
  {\seq {\para {\fam i {\setref {x_i} {\comp {t_i}}}}} {\comp u}}
\end{mathpar}
To compile $\letrecin {\fam i {x_i = t_i}} u$, we create
uninitialized store cells for each
$x_i$,\footnote{$(\newin {\fam {i \in I} {x_i}} \square)$ denotes a sequence of $(\newin {x_i} \square)$ binders in an
  arbitrary order. In particular, $(\newin \emptyset t)$ is just $t$.}
then we compute the assignments $\setref {x_i} {\comp {t_i}}$ in an
arbitrary order, and finally we evaluate $\comp u$. Note that all the
$x_i$ are in the scope of each $\comp {t_j}$: the translation respects
the scoping of the $\letrec {\fam i {x_i = t_i}} {}$ construct.

\begin{mathparfig}{fig:compilation-into-store}{Compiling \texttt{letrec} into store updates}
  \begin{array}{rll}
    \comp x
    & \defeq &
    x
    \\
    \comp {\lam x t}
    & \defeq &
    \lam x {\comp t}
  \end{array}

  \begin{array}{rll}
    \comp {\app t u}
    & \defeq &
    \app {\comp t} {\comp u}
    \\
    \comp {\constr K {\fam i {t_i}}}
    & \defeq &
    \constr K {\fam i {\comp {t_i}}}
  \end{array}

  \begin{array}{rll}
    \comp {\match t {\fam i {\clause {p_i} {t_i}}}}
    & \defeq &
    \match {\comp t} {\fam i {\clause {p_i} {\comp {t_i}}}}
    \\
    \comp {\letrecin {\fam i {x_i = t_i}} u}
    & \defeq &
    \newin {\fam i {x_i}}
      {\seq {\para {\fam i {\setref {x_i} {\comp {t_i}}}}} {\comp u}}
  \end{array}
\end{mathparfig}

\subsection{Relating Target Terms back to Source Terms}

\begin{mathparfig}{fig:head-term-relation}{Head term relation $\termrel t {\conf H {t'}}$}
  \fbox{Simple rules} \hfill\\

  \infer
  { }
  {\termrel x {\conf \emptyset x}}

  \infer
  { }
  {\termrel {\lam x t} {\conf \emptyset {\lam x {\comp t}}}}

  \infer
  {\termrel t {\conf {H_t} {t'}}
   \\
   \termrel u {\conf {H_u} {u'}}}
  {\termrel
    {\app t u}
    {\conf {\heapjoin {H_t} {H_u}} {\app {t'} {u'}}}}

  \infer
  {\fam {i \in I} {\termrel {t_i} {\conf {H_i} {t'_i}}}}
  {\termrel
    {\constr K {\fam {i \in I} {t_i}}}
    {\conf {\fam {i \in I} {H_i}}
           {\constr K {\fam {i \in I} {t'_i}}}}
  }

  \infer
  {\termrel t {\conf H {t'}}}
  {\termrel
    {\match t h}
    {\conf H {\match {t'} {\comp h}}}}

  \\\fbox{$\kwd{letrec}$ rules} \hfill\\

  \infer[init]
  {J \neq \emptyset}
  {
    {\letrecin
      {\fam {i \in I} {x_i = t_i}, \fam {j \in J} {x_j = t_j}}
      u}
    \\ \termrel {}
    {\conf {\heapextfam {} {i \in I} {x_i} \bot}
      {\newin {\fam {j \in J} {x_j}}
        {\seq {\para {\fam {k \in I \uplus J} {\setref {x_k} {t_k}}}} {\comp u}}}
    }
  }

  \infer[write]
  {\fam {i \in I} {\termrel {v_{s,i}} {\conf {B_i} {v_{t,i}}}}
   \\
   \fam {j \in J} {\termrel {t_{s,j}} {\conf {H_j} {t_{t,j}}}}
  }
  {
   {\letrecin {\fam {i \in I} {x_i = v_{s,i}}, \fam {j \in J} {y_j = {t_{s,j}}}} u}
   \\ \termrel {}
    {\conf
      {\heapjoin
        {\heapextfam
          {\heapextfam {} {i \in I} {x_i} {v_{t,i}}}
          {j \in J} {y_j} \bot}
        {\heapjoin {\fam {i \in I} {B_i}} {\fam {j \in J} {H_j}}}}
      {\seq {\para {(\fam {i \in I} \Done, \fam {j \in J} {\setref {y_j} {t_{t,j}}})}} {\comp u}}
    }}

  \infer[done]
  {\fam {i \in I} {\termrel {v_{s,i}} {\conf {B_i} {v_{t,i}}}}}
  {\termrel
    {\letrecin {\fam {i \in I} {x_i = v_{s,i}}} u}
    {\conf
      {\heapjoin
        {\heapextfam {} {i \in I} {x_i} {v_{t,i}}}
        {\fam {i \in I} {B_i}} H}
      {\seq \Done {\comp u}}
    }
  }

  \infer[further]
  {\fam {i \in I} {\termrel {v_{s,i}} {\conf {B_i} {v_{t,i}}}}
   \\
   \termrel {u_s} {\conf H {u_t}}
  }
  {\termrel
    {\letrecin {\fam {i \in I} {x_i = v_{s,i}}} {u_s}}
    {\conf
      {\heapjoin
        {\heapextfam {} {i \in I} {x_i} {v_{t,i}}}
        {\heapjoin {\fam {i \in I} {B_i}} H}}
      {u_t}
    }
  }

  \\\fbox{Heap rules} \hfill\\

  \infer[heap-weaken]
  {\termrel {t_s} {\conf H {t_t}}
   \\
   \forall x \in \dom{B},\ x \notin \conf H {t_t}}
  {\termrel {t_s}
    {\conf {\heapjoin H B} {t_t}}}

  \newcommand{\twolines}[2]{{\begin{array}{c}#1 \\ #2\end{array}}}
  \infer[heap-copy]
  {\twolines{}{\termrel {t_s} {\conf H {t_t}}}
   \\
   \twolines
     {\compatible \phi H}
     {\dom \phi \subseteq \dom H}
  }
  {\termrel {t_s} {\conf {\phi(H)} {\phi(t_t)}}}
\end{mathparfig}

We want to prove that this translation scheme is safe; that source
terms that do not go vicious translate into target terms that do not
segfault. This requires a backward simulation property: any reduction
path from a translation in the target (in particular, a reduction path
that leads to a segfault) needs to be simulated by a reduction path on the
original source term (in particular, a reduction path that goes vicious).

We define a \emph{head term relation} $\termrel t {\conf H {t'}}$ that
relates the parts of a term and a configuration that are in reducible
position. For non-reducible positions, we use the direct embedding
$\comp t$ above. In this relation, $H$ contains not all the locations
that are bound in $t'$, but only those that correspond to
$\kwd{letrec}$-binding found in $t$; intuitively, $H$ is the disjoint union
of all the $\kwd{letrec}$-bindings in $t$, seen as local store fragments.

The rules for the $\kwd{letrec}$ constructs correspond to
a decomposition of the various intermediate states of the reduction of
their backpatching compilation.

Finally, the ``heap rules'' give reasoning principles to bridge the
difference between the way the local and global stores evolve during
reduction. In our source-level semantics, explicit substitutions
(local store) may be duplicated or erased during reductions involving
the values they belong to. Store duplication is expressed by
a variable-renaming substitution $\phi$, that ``merges'' different
fragments of the global store together; the side-condition
$(\compatible \phi H)$ guarantees that the resulting store is
well-formed.

For reasons of space, we moved the explanation of this relation,
including the definition of $(\compatible \phi H)$, as well as the
proofs of its properties, to
\myappendixfullref{ann:simulation-proof}. The two key results are
mentioned here, guaranteeing that our analysis is also sound for
this global-store semantics.

\begin{theorem}[Backward Simulation]\label{thm:backward-simulation}~\\
If
\begin{mathline}
  \termrel {t_s} {\conf H {t_t}}

  \rew {\conf H {t_t}} {\conf {H'} {t'_t}}
\end{mathline}
then $\exists t'_s$,
\begin{mathline}
  \rewopt {t_s} {t'_s}

  \termrel {t'_s} {\conf {H'} {t'_t}}
\end{mathline}
\end{theorem}

\begin{theorem}[$\Return$-typed programs cannot segfault]\label{thm:no-segfault}
\begin{mathline}
  \der \emptyset {t_s} \Return

  \wedge

  \rewstar {\conf \emptyset {\comp  {t_s}}} {\conf H {t'_t}}

  \implies

  {\conf H {t'_t}} \notin \mathsf{Segfault}
\end{mathline}
\end{theorem}
\end{version}

\myfullref{Theorem}{thm:no-segfault} guarantees that our analysis is
sound for both our source language and its backpatching
translation. \myfullref{Theorem}{thm:backward-simulation} also tells
us that our source semantics has ``enough'' reduction rules compared
to the global-store semantics. For example, if the global-store
semantics computes a value for a term, then the source semantics would
have computed a related value. %% This should reassure readers who are
%% less familiar with explicit-substitution semantics.

\section{Extension to a Full Language}
\label{section:extension-full}

We now discuss the extension of our typing rules to the full OCaml
language, whose additional features (e.g.~exceptions, first-class
modules and GADTs) contain subtleties that need special care.

\subsection{The Size Discipline}
\label{subsec:size}

The OCaml compilation scheme, one of several possible ways of treating
recursive declarations, proceeds by reserving heap blocks for the
recursively-defined values, and using the addresses of these heap
blocks (which will eventually contain the values) as dummy values: it
adds the addresses to the environment and computes the values
accordingly. If no vicious term exists, the addresses are never
dereferenced during evaluation, and evaluation produces ``correct''
values. Those correct values are then moved into the space occupied by
the dummies, so that the original addresses contain the correct
result.

This strategy depends on knowing how much space to allocate for each
value. Not all OCaml types have a uniform size; e.g.~variant types
may contain constructors with different arities, resulting
in different in-memory sizes, and the size of a closure depends on the
number of free variables.

After checking that mutually-recursive definitions are meaningful
using the rules we described, the OCaml compiler checks that it can
realize them, by trying to infer a static size for each value. It then
accepts to compile each declaration if either:
\begin{itemize}
\item it has a static size, or
\item it doesn't have a statically-known size, but its usage mode of
  mutually-recursive definitions is always Ignore
\end{itemize}
(The second category corresponds to detecting some values that are
actually non-recursive and lifting them out. Such non-recursive values
often occur in standard programming practice, when it is more
consistent to declare a whole block as a single \code{let rec} but
only some elements are recursive.)

This static-size test may depend on lower-level aspects of
compilation, or at least value representation choices. For example,
\begin{lstlisting}
   if p then (fun x -> x) else (fun x -> not x)
\end{lstlisting}
has a static size (both branches have the same size), but
\begin{lstlisting}
   if p then (fun x -> x + 1) else (fun x -> x + offset)
\end{lstlisting}
does not: the second function depends on a free variable
\lstinline|offset|, so it will be allocated in a closure with an extra
field. (While \lstinline|not| is also a free variable, it is
statically resolvable to a global name.)

\paragraph{Relation to the mode system}

The mode system corresponds to a correctness criterion on the
operational semantics of programs; it is independent of compilation
schemes.
In contrast, the size discipline corresponds to a restrictive
compilation strategy for value recursion that involves rejecting
certain definitions.
The size discipline is formalized
by~\citet{hirschowitz-compilation-2009}; it would be possible to
incorporate it into our system, modelling it as a separate judgment to be
checked for well-moded definitions (rather than as an enrichment of
the mode judgment).
However, the resulting system would be less portable to programming
languages whose value representations differ from OCaml's, and which
consequently would not use the same size discipline.

\subsection{Dynamic Representation Checks: Float Arrays}% and lazy thunk forwarding}
\label{section:float-arrays}

OCaml uses a dynamic representation check for its polymorphic arrays:
when the initial array elements supplied at array-creation time are
floating-point numbers, OCaml chooses a specialized, unboxed
representation for the array.

Inspecting the representation of elements during array creation means
that although array construction looks like a guarding context, it is
often in fact a dereference.  There are three cases to consider:
first, where the element type is statically known to be
\lstinline{float}, array elements will be unboxed during creation,
which involves a dereference;
second, where the element type is statically known not to be
\lstinline{float}, the inspection is elided;
third, when the element type is not statically known the elements will
be dynamically inspected --- again a dereference.

The following program must be rejected, for example:
\begin{lstlisting}
let rec x = (let u = [|y|] in 10.5)
    and y = 1.5
\end{lstlisting}
since creating the array \code{ [|y|] } will unbox the element \code{y},
leading to undefined behavior if \code{y} --- part of the same recursive
declaration --- is not yet initialized.

\subsection{Exceptions and First-Class Modules}
\label{section:exceptions}

In OCaml, exception declarations are generative: if a functor body contains
an exception declaration then invoking the functor twice will declare two
exceptions with incompatible representations, so that catching one of them
will not interact with raising the other.

Exception generativity is implemented by allocating a memory cell at
functor-evaluation time (in the representation of the resulting
module); and including the address of this memory cell as an argument
of the exception payload. In particular, creating an exception value
\lstinline|M.Exit 42| may dereference the module \lstinline|M| where
\lstinline|Exit| is declared.

Combined with another OCaml feature, first-class modules, this
generativity can lead to surprising incorrect recursive declarations,
by declaring a module with an exception and using the exception in the
same recursive block.

For instance, the following program is unsound and rejected by our
analysis:
\begin{lstlisting}
module type T = sig exception A of int end
let rec x = (let module M = (val m) in M.A 42)
and (m : (module T)) = (module (struct exception A of int end) : T)
\end{lstlisting}
In this program, the allocation of the exception value \code{M.A 42}
dereferences the memory cell generated for this exception in the
module \code{M}; but the module \code{M} is itself defined as the
first-class module value \code{(m : (module T))}, part of the same
recursive nest, so it may be undefined at this point.

(This issue was first
\href{https://github.com/ocaml/ocaml/pull/556#issuecomment-271159296}{pointed
  out} by Stephen Dolan.)

\subsection{GADTs}
\label{section:gadts}

The original syntactic criterion for OCaml was implemented not
directly on surface syntax, but on an intermediate representation
quite late in the compiler pipeline (after typing, type-erasure, and
some desugaring and simplifications).
In particular, at the point where the check took place,
exhaustive single-clause matches such as  \lstinline|match t with x -> $\ldots$|
or \lstinline|match t with () -> $\ldots$|) had been transformed into direct
substitutions.

This design choice led to programs of the following form being accepted:
\begin{lstlisting}
  type t = Foo
  let rec x = (match x with Foo -> Foo)
\end{lstlisting}
\noindent
While this appears innocuous, it becomes unsound with the
addition of GADTs to the language:
\begin{lstlisting}
  type (_, _) eq = Refl : ('a, 'a) eq
  let universal_cast (type a) (type b) : (a, b) eq =
    let rec (p : (a, b) eq) = match p with Refl -> Refl in p
\end{lstlisting}
For the GADT |eq|, matching against \lstinline{Refl} is not a no-op:
it brings a type equality into scope that expands the set of
types that can be assigned to the
program~\citep{DBLP:conf/aplas/GarrigueR13}.
It is therefore necessary to treat matches involving GADTs as
inspections to ensure that a value of the appropriate type is actually
available; without that change definitions such as
\lstinline{universal_cast} violate type safety.

\subsection{Laziness}

OCaml's evaluation is eager by default, but it supports an explicit
form of lazy evaluation: the programmer can write |lazy e| and |force e| to delay and force the evaluation of an expression.

The OCaml implementation performs a number of optimizations involving
|lazy|.  For example, when the argument of |lazy| is a trivial
syntactic value (variable or constant) for which eager and lazy
evaluation behave equivalently, the compiler eliminates the
allocation of the lazy thunk.

However, for recursive definitions eager and lazy evaluation are not
equivalent, and so the recursion check must treat |lazy trivialvalue|
as if the |lazy| were not there.
For example, the following recursive definition is disallowed, since
the optimization described above nullifies the delaying effect of the
|lazy|
\begin{lstlisting}
   let rec x = lazy   y   and y = $\ldots$
\end{lstlisting}
while the following definition is allowed by the check, since the
argument to |lazy| is not sufficiently trivial to be subject to the
optimization:
\begin{lstlisting}
   let rec x = lazy (y+0) and y = $\ldots$
\end{lstlisting}

Our typing rule for \lstinline|lazy| takes this into account:
``trivial'' thunks are checked in mode $\Return$ rather than $\Delay$.

\section{Related Work}
\label{subsec:related-work}

%% \paragraph{Backward analyses} Our right-to-left reading is
%% a particular case of \emph{backward analysis}, as presented for
%% example by \citet*{backward-analysis}. A lot of work on backward
%% analysis for functional programs has a denotational flavor, while we
%% stick to a type system, giving a more declarative
%% presentation.

\paragraph{Degrees} \citet*{boudol-objects-2001} introduces the
notion of ``degree'' $\alpha \in \{0, 1\}$ to statically analyze
recursion in object-oriented programs (recursive objects,
lambda-terms). Degrees refine a standard ML-style type system for
programs, with a judgment of the form $\Gamma \vdash t : \tau$ where
$\tau$ is a type and $\Gamma$ gives both a type and a degree for each
variable. A context variable has degree $0$ if it is required to
evaluate the term (related to our $\Dereference$), and $1$ if it is
not required (related to our $\Delay$). Finally, function types are
refined with a degree on their argument: a function of type
$\tau^0 \to \tau'$ accesses its argument to return a result, while
a $\tau^1 \to \tau'$ function does not use its argument right away,
for example a curried function $\lam x {\lam y {\pair x y}}$ --- whose
argument is used under a delay in its body $\lam y {\pair
  x y}$. Boudol uses this reasoning to accept a definition such as
\code{let rec obj = class_constructor obj params},
arising from object-oriented encodings, where
\code{class_constructor}
has a type $\tau^0 \to \dots$.

Our system of mode is finer-grained than the binary degrees of Boudol;
in particular, we need to distinguish $\Dereference$ and $\Guard$ to
allow cyclic data structure constructions.

\sloppy
On the other hand, we do not reason about the use of function
arguments at all, so our system is much more coarse-grained in this
respect. In fact, refining our system to accept
\code{let rec obj = constr obj params}
would be incorrect for our use-case in the OCaml compiler, whose
compilation scheme forbids passing yet-uninitialized data to
a function.

\fussy
In a general design aiming for maximal expressiveness, access modes should
refine ML types; in Boudol's system, degrees are interlinked with the
type structure in function types $\tau^\alpha \to \tau'$, but one
could also consider pair types of the form
${\tau_1}^{\alpha_1} \times {\tau_2}^{\alpha_2}$, etc. In our simpler
system, there are no interaction between value shapes (types) and
access modes, so we can forget about types completely, a nice
conceptual simplification. Our formalization is done entirely in
an untyped fragment of ML.

\paragraph{Compilation}

\citet*{hirschowitz-compilation-2003,hirschowitz-compilation-2009}
discuss the space of compilation schemes for recursive value
definitions. Their work is an inspiration for our own compilation
result: they provide a source-level semantics based on floating
bindings upwards in the term (similar to explicit substitutions or
local thunk stores), and prove correctness of compilation to a global
store with backpatching.

Our source-level semantics is close to theirs in spirit (we would
argue that the use of \emph{reduction at a distance} is
an improvement), and our compilation scheme and its correctness proof
are not novel compared to their work --- they are there to provide
additional intuition. The main
contribution of our work is our mode system for recursive 
declarations, which is expressive enough to capture OCaml value
definitions, yet simple and easy to infer.

A natural question for our work is whether the access-mode derivations
we build in our safety check can inform the compilation strategy for
recursive values. We concentrate on safety, but there is
\href{https://github.com/ocaml/ocaml/pull/8956}{ongoing work} by others
on the compilation method for recursive values, which partly
goes in this direction.

\paragraph{Fixing Letrec (Reloaded)}
Fixing Letrec (Reloaded)~\citep*{fixing-letrec,fixing-letrec-reloaded}
is a nice brand of work from the Scheme community, centered on
producing efficient code for recursive value declarations,
even in presence of dynamic checks for the absence of
uninitialized-name reads. It presents a static analysis, both for
optimization purposes (eliding dynamic safety checks) and user
convenience. The analysis is described by informal prose, but it is
similar in spirit to our mode analysis (using modes named ``protected'',
which corresponds to our $\Delay$, ``protectable'' which sounds like
$\Return$ and ``unsafe'' which is $\Dereference$). We precisely describe
an analysis (richer, as it also has a $\Guard$ mode) and prove its
correctness with respect to a dynamic semantics.

\paragraph{Name access as an effect} \citet*{dreyer-2004} proposes to
track usage of recursively-defined variables as an effect, and designs
a type-and-effect system whose effects annotations are sets of
abstract names, maintained in one-to-one correspondence with
\code{let rec}-bound variables. The construction \code{let rec
  X$\triangleright$ x : $\;\tau\;$ = e} introduces the abstract
type-level name \code{X} corresponding to the recursive variable
\code{x}. This recursive variable is made available in the scope of
the right-hand-side \code{e : $\;\tau$} at the type
\code{box(X,$\tau$)} instead of $\tau$ (reminding us of
guardedness modalities). Any dereference of \code{x} must
explicitly ``unbox'' it, adding the name \code{X} to the ambient
effect.

This system is very powerful, but we view it as a core language rather
than a surface language: encoding a specific usage pattern may require
changing the types of the components involved, to introduce explicit
\code{box} modalities:
\begin{itemize}
\item When one defines a new function from $\tau$ to $\tau'$, one
  needs to think about whether it may be later used with
  still-undefined recursive names as argument --- assuming it indeed
  makes delayed uses of its argument. In that case, one should use the
  usage-polymorphic type function type
  $\forall X. \mathsf{box}(X, \tau) \to \tau'$ instead of the simple
  function type $\tau \to \tau'$. (It is possible to inject $\tau$
  into $\mathsf{box}(X, \tau)$, so this does not restrict
  non-recursive callers.)
\item One could represent cyclic data such as
\code{let rec ones = 1 :: ones}
in this system, but it would require a non-modular change of the type
of the list-cell constructor from
$\forall \alpha. \alpha \to \mathsf{List}(\alpha) \to \mathsf{List}(\alpha)$
to the box-expecting type
$\forall \alpha. \alpha \to \forall X. \mathsf{box}(X,\mathsf{List}(\alpha)) \to \mathsf{List}(\alpha)$
.

\end{itemize}

In particular, one cannot directly use typability in this system as
a static analysis for a source language; this work needs to be
complemented by a static analysis such as ours, or the safety has to
be proved manually by the user placing box annotations and operations.
However, we believe that any well-typed
program that is accepted by our mode system could be encoded in
Dreyer's system,
%%  (and in practice our check runs after OCaml's
%% type-checker, so for our use-case does work with well-typed programs –
%% in a more powerful type system than Dreyer's).
%
roughly as follows:

\begin{itemize}
\item
if in the context $\Gamma$ of a derivation $\Gamma \vdash t : \Return$
  in our system, we have $x : \Dereference$, then in the encoding the
  corresponding effect variable $X$ would be an ambient capability
 ($\Gamma \vdash t : \tau [T]$ with $X \in T$)
\item
 on the other hand, if we have $x : \Return$ or a more permissive
  mode, then we would give the corresponding term variable in the
  encoding type $\mathsf{box}(X,T)$
\end{itemize}

So, for example,
$x: \Dereference \vdash x + 1 : \Return$
would be encoded as
$X, x \vdash x + 1 : \text{Int} [X]$
but
$x : \Guard \vdash \{ t = x \} : \Return$
would be encoded as
$X, x:\mathsf{box}(X,\tau) \vdash \{ t = x \} : \{ t : \tau \} [\emptyset]$.

The whole derivation of an encoding of a valid recursive definition would have
a non-boxed type on the right-hand side, %%  (so $x:\Return \vdash x:\Return$
%% does not give a valid derivation without unboxing),
without any of the effect variables of the recursively-defined
variable in the ambient context.

\paragraph{Strictness analysis} Our analysis can be interpreted as
a form of strictness or demand analysis, with modes above $\Return$
being non-strict and $\Dereference$ being the forcing mode. Note
however that we use a may-analysis (a $\Dereference$ variable
\emph{may} be dereferenced) while strictness optimizations usually
rely on a must-analysis (we only give $\mathsf{Forcing}$ when we know
for sure that forcing happens); to do this one should change our
interpretation of unknown functions to be conservative in the other
direction, with mode $\Ignore$ rather than $\Dereference$ for their
arguments. More importantly, strictness analyses typically try to
compute more information than our modes: besides the question of
whether a given subterm will be forced or not, they keep track of which
prefixes of the possible term shapes will get forced --- this is more
related to the finer-grained spaces of ``ranks'' or ``degrees'' for
recursive functors.

\paragraph{Graph typing} Hirschowitz also collaborated on static
analyses for recursive definitions in
\citet*{hirschowitz-graph-2005,bardou-2005}. The design goal was
 a simpler system than existing work aiming for expressiveness, with
 inference as simple as possible.

As a generalization of Boudol's binary degrees they use compactified
numbers $\mathbb{N}\cup\{-\infty, \infty\}$. The degree of a free
variable ``counts'' the number of subsequent $\lambda$-abstractions
that have to be traversed before the variable is used; $x$ has degree
$2$ in $\lam y {\lam z x}$. A $-\infty$ is never safe, it corresponds
to our $\Dereference$ mode. $0$ conflates our $\Guard$ and $\Return$
mode (an ad-hoc syntactic restriction on right-hand-sides is used to
prevent under-determined definitions), the $n+1$ are fine-grained
representations of our $\Delay$ mode, and finally $+\infty$ is our
$\Ignore$ mode.

Another salient aspect of their system is the use of ``graphs'' in the
typing judgment: a use of \code{y} within a definition
\code{let x = e}
is represented as an edge from \code{y} to \code{x} (labeled by the
usage degree), in a constraint graph accumulated in the typing
judgment. The correctness criterion is formulated in terms of the
transitive closure of the graph: if \code{x} is later used somewhere,
its usage implies that \code{y} also needs to be initialized in this
context.

One contribution of our work is to show that a more standard syntactic
approach can replace the graph representation. Note that our typing
rule for mutual-recursion uses a fixpoint computation, reminiscent of
their transitive-closure computation but within a familiar type-system
presentation.

Finally, their static analysis mentions the in-memory size of values,
which needs to be known statically, in the OCaml compilation scheme,
to create uninitialized memory blocks for the recursive names before
evaluating the recursive definitions. Our mode system does not mention
size at all, it is complemented by an independent (and simpler)
analysis of static-size deduction, which is outside the scope of the
present formalization, but described briefly in
\myfullref{Section}{subsec:size}.

\paragraph{\Fsharp} \citet*{syme-2006} proposes a simple translation
of mutually-recursive definitions into delay and force constructions
that introduce and eliminate lazy values.
For
example, \lstinline|let rec $x\;$=$\;t$ and $y\;$=$\;u$| is turned into the following:
\begin{lstlisting}
let rec $x_{\mathsf{thunk}}$ = lazy ($t[\mathsf{force}~x_{\mathsf{thunk}} / x,
                                      \mathsf{force}~y_{\mathsf{thunk}} / y]$)
     and $y_{\mathsf{thunk}}$ = lazy ($u[\mathsf{force}~x_{\mathsf{thunk}} / x,
                                       \mathsf{force}~y_{\mathsf{thunk}} / y]$)
let $x$ = force $x_{\mathsf{thunk}}$ and $y$ = force $y_{\mathsf{thunk}}$
\end{lstlisting}

With this semantics, evaluation happens on-demand, which the recursive
definitions evaluated at the time where they are first accessed. This
implementation is very simple, but it turns vicious definitions into
dynamic failures --- handled by the lazy runtime which safely raises an
exception. However, this elaboration cannot support cyclic data
structures: The translation of \code{let rec ones = 1 :: ones} fails
at runtime:
\begin{lstlisting}
let rec ones$_{\mathsf{thunk}}$ = lazy (1 :: force ones$_{\mathsf{thunk}}$)
\end{lstlisting}
Furthermore, the translation affects the semantics of programs in
surprising ways: in particular, the implicit introduction of laziness
into definitions that start new threads can lead to unexpected
multiple execution of computations and to race conditions.

Nowadays, \Fsharp{} provides an ad-hoc syntactic criterion, the
``Recursive Safety
Analysis'' \citep*{fsharp-recursive-safety-analysis}, roughly similar
to the previous OCaml syntactic criterion, that distinguishes ``safe''
and ``unsafe'' bindings in a mutually-recursive group; only the latter
are subjected to the thunk-introducing translation.

Finally, the implementation also performs a static analysis to detect
some definitions that are bound to fail --- it over-approximates
safety by ignoring occurrences within delaying terms (function
abstractions or objects or lazy thunks) even if those delaying terms
may themselves be used (i.e.~respectively called or accessed or
forced) at definition time.  We believe that we could recover a
similar analysis by changing our typing rules for our constructions
--- but with the OCaml compilation scheme we must absolutely remain
sound.

\paragraph{Needed computations}

Further connections between laziness and recursive call-by-value
definitions may be drawn: for example,
\citet{DBLP:conf/esop/ChangF12} characterize  
call-by-need by introducing \emph{needed computations},
which are similar in spirit to our idea of \emph{forcing contexts}.
%
%% Letrec and laziness have commonalities and similar concepts show up
%% (including the benefits of using explicit-substitution over a global
%% store). In practice the two definitions are similar in spirit, but
%% different.
%
However, the set of computations characterized by the two ideas are
different in practice: for example, in our system \code{f (fst -)} is
a forcing context, but the hole is not in ``needed position'' for
call-by-need.

Intuitively,
needed computations correspond to
  positions at which \emph{any} choice of reduction order
  \emph{must force} the computation to make progress,
while
forcing contexts correspond to
 positions where \emph{some} choice of reduction order
 \emph{may} force the computation (and fail if the value is initialized).

\paragraph{Operational semantics}

\citet*{felleisen-local-store} and \citet{ariola-felleisen} propose local-store
semantics (for a call-by-value store and a call-by-need thunk
store, respectively) that can express recursive bindings. The
source-level operational semantics of
\citet*{hirschowitz-compilation-2003,hirschowitz-compilation-2009} is
more tailored to recursive bindings, manipulated as explicit
substitutions, although the relation to standard explicit-substitution
calculi is not made explicitly. They also provide a global-store
semantics for their compilation-target language with mutable
stores. \citet*{boudol-abstract-machine-2002} and \citet*{dreyer-2004}
use an abstract machine. \citet*{syme-2006} translates recursive
definitions into lazy constructions, so the usual thunk-store
semantics of laziness can be used to interpret recursive
definitions. Finally, \citet*{dynamic-semantics-2008} give the
simplest presentation of a source-level semantics we know of; we
extend it with algebraic datatypes and pattern-matching, and use it as
a reference to prove the soundness of our analysis.

Our own experience presenting this work is that local-store semantics
has been largely forgotten by the programming-language community,
which is a shame as it provides a better treatment of recursive
definitions (or call-by-need) than global-store semantics.

One inessential detail in which the semantics often differ is the
evaluation order of mutually-recursive right-hand-sides. Many
presentations enforce an arbitrary (e.g.~left-to-right) evaluation
order. Some systems~\citep*{syme-2006,dynamic-semantics-2008} allow
a reduction to block on a variable whose definition is not yet
evaluated, and go evaluate it in turn; this provides the ``best
possible order'' for the user. Another interesting variant would be to
say that the reduction order is unspecified, and that an uninitialized variable
is a stuck term whose evaluation results in a fatal error; this
provides the ``worst possible order'', failing as much as possible; as
far as we know, the previous work did not propose it, although it is
a simple presentation change. Most static analyses are
evaluation-order-independent, so they are sound and complete with
respect to the ``worst order'' interpretation.

\section{Conclusion}

We have presented a new static analysis for recursive value
declarations, designed to solve a fragility issue in the OCaml
language semantics and implementation.
It is less expressive than previous works that analyze function calls
in a fine-grained way; in return, it remains fairly simple, despite
its ability to scale to a fully-fledged programming language, and the
constraint of having a direct correspondence with a simple inference
algorithm.

We believe that this static analysis may be of use for other
functional programming languages, both typed and untyped.
It also seems likely that the techniques we have used in this work will
apply to other systems --- type parameter variance, type constructor
roles, and so on.
Our hope in carefully describing our system is that we will
eventually see a pattern emerge for the design and structure of
``things that look like type systems'' in this way.

%% \newpage
%% \begin{small} % try to fit in one page
\bibliography{letrec}
%% \end{small}

\newpage
\appendix

\section{Properties of our Typing Judgment}
\label{ann:properties}

The following technical results can be established by simple
inductions on typing derivations, without any reference to an
operational semantics.

\begin{lemma}[$\Ignore$ inversion]\label{lem:invert-ignore}
  $\der \Gamma t \Ignore$ is provable with only $\Ignore$ in
  $\Gamma$.
\end{lemma}

\begin{lemma}[$\Delay$ inversion]\label{lem:invert-delay}
  $\der \Gamma t \Delay$ holds exactly when
  $\Gamma$ maps all free variables of $t$ to $\Delay$ or $\Ignore$.
\end{lemma}

\begin{lemma}[$\Dereference$ inversion]\label{lem:invert-dereference}
  $\der \Gamma t \Dereference$ holds exactly when
  $\Gamma$ maps all free variables of $t$ to $\Dereference$.
\end{lemma}

\begin{lemma}[Environment flow]
  \label{lem:environment-flow}
  If a derivation $\der \Gamma t m$ contains a sub-derivation $\der {\Gamma'} {t'} {m'}$, then
  $
    \forall x \in \Gamma,
    \;
    \Gamma(x) \succeq \Gamma'(x)
  $.
\end{lemma}

\begin{lemma}[Weakening]
  \label{lem:weakening}
  If $\der \Gamma t m$ holds then $\der {\envsum \Gamma {\Gamma'}} t m$ also holds.
\end{lemma}
(Weakening would not be admissible if our variable rule imposed
$\Ignore$ on the rest of the context.)

\begin{lemma}[Substitution]
  \label{lem:substitution}
  If $\der {\Gamma, x : {m_u}} t m$ and $\der {\Gamma'} u {m_u}$ hold,
  then $\der {\envsum \Gamma {\Gamma'}} {\subst t {\by u x}} m$ holds.
\end{lemma}

\begin{lemma}[Subsumption elimination]
  Any derivation in the system can be rewritten so that the
  subsumption rule is only applied with the variable rule as premise.
\end{lemma}

\begin{theorem*}[\outlineref{thm:principal-environments}]
  Whenever both $\der {\Gamma_1} t m$ and $\der {\Gamma_2} t m$ hold,
  then $\der {\min(\Gamma_1, \Gamma_2)} t m$ also holds.
\end{theorem*}

\begin{proof}
  The proof first performs subsumption elimination on both
  derivations, and then by simultaneous induction on the results. The
  elimination phase makes proof syntax-directed, which guarantees that
  (on non-variables) the same rule is always used on both sides in
  each derivation.
\end{proof}

This results tells us that whenever $\der \Gamma t m$ holds, then it
holds for a minimal environment $\Gamma$ --- the minimum of all
satisfying $\Gamma$.

\begin{definition}[Minimal environment]
  $\Gamma$ is minimal for $t : m$ if $\der \Gamma t m$ and,
  for any $\der {\Gamma'} t m$ we have $\Gamma \preceq \Gamma'$.
\end{definition}

In fact, we can give a precise characterization of ``minimal''
derivations, that uniquely determines the output of our backwards analysis
algorithm.

\begin{definition}[Minimal binding rule]
  An application of the binding rule is \emph{minimal} exactly when
  the choice of $\Gamma'_i$ is the least solution to the recursive
  equation in its third premise.
\end{definition}

\begin{definition}[Minimal derivation]
  A derivation is \emph{minimal} if it does not use the subsumption
  rule, each binding rule is minimal and, in the conclusion
  $\der \Gamma x m$ of each variable rule, $\Gamma$ is minimal for
  $x : m$.
\end{definition}

\begin{definition}[Minimization]
  Given a derivation $\deriv D {\der \Gamma t m}$, we define the
  (minimal) derivation $\minimal{D}$ by:
  \begin{itemize}
  \item Turning each binding rule into a minimal version of this
    binding rule --- this may require applying
    \myfullref{Lemma}{lem:weakening} to the
    \code{let rec} derivation below.
  \item Performing subsumption-elimination to get another
    derivation of $\der \Gamma t m$.
  \item Replacing the context of each variable rule by the minimal
    context for this variable --- this does not introduce new
    subsumptions.
  \end{itemize}
\end{definition}

\begin{fact}[Minimality]
  If $\deriv D {\der \Gamma t m}$ and $\deriv {\minimal D} {\der {\Gamma_m} t m}$,
  then $\Gamma_m \preceq \Gamma$.
\end{fact}

\begin{lemma}[Stability]\label{lem:stability}
  If $D$ is a minimal derivation, then $\minimal{D} = D$.
\end{lemma}

\begin{lemma}[Determinism]
  \label{lem:determinism}
  If $\deriv {D_1} {\der {\Gamma_1} t m}$
  and $\deriv {D_2} {\der {\Gamma_2} t m}$,
  then $\minimal{D_1}$ and $\minimal{D_2}$
  are the same derivation.
\end{lemma}

\begin{corollary}[Minimality equivalence]\label{cor:minimality-equivalence}
  The environment $\Gamma$ of a derivation $\der \Gamma t m$ is
  minimal for $t : m$ if and only if $\der \Gamma t m$ admits
  a minimal derivation.
\end{corollary}

\begin{proof}
  If $\Gamma$ is minimal for $t : m$, then the context
  $\Gamma_m \preceq \Gamma$ obtained by minimization must itself be
  $\Gamma$.

  Conversely, if a derivation $\deriv {D_m} {\der \Gamma t m}$ is
  minimal, then all other derivations $\der {\Gamma'} t m$ have $D_m$
  as minimal derivation by \myfullref{Lemma}{lem:stability} and
  \myfullref{Lemma}{lem:determinism}, so $\Gamma \preceq \Gamma'$ holds.
\end{proof}

\begin{theorem}[Localization]
  \label{thm:localization}
  $\der \Gamma t {m'}$ implies $\der {\mcomp m \Gamma} t {\mcomp m {m'}}$.

  Furthermore, if $\Gamma$ is minimal for $t : {m'}$, then
  $\mcomp m \Gamma$ is minimal for $t : \mcomp m {m'}$.
\end{theorem}

\begin{proof}
  The proof proceeds by direct induction on the derivation, and does
  not change its structure: each rule application in the source
  derivation becomes the same source derivation in the result. In
  particular, minimality of derivations is preserved, and thus, by
  \myfullref{Corollary}{cor:minimality-equivalence}, minimality of
  environments is preserved.

  Besides associativity of mode composition, many cases rely on the
  fact that external mode composition preserves the mode order
  structure: $m'_1 \prec m'_2$ implies
  $\mcomp m {m'_1} \prec \mcomp m {m'_2}$, and
  $\max(\mcomp m {m'_1}, \mcomp m {m'_2})$ is
  $\mcomp m {\max(m'_1, m'_2)}$.
\end{proof}

% acmart uses MakeTextUppercase, some versions of which try to
% uppercase reference names as well, turning subsec:soundness into
% SUBSEC:SOUNDNESS. Fixed by \NoCaseChange.
\section{Proofs for \NoCaseChange{\myfullref{Section}{subsec:soundness}}}
\label{ann:proofs}

% Ef ::= L | E[Ff[L]]

\begin{lemma}[\outlineref{lem:forcing-modes}]~\\
  If $\der {\Gamma, \of x {m_x}} {\plug {\forcing E} x} m$ with $m \succeq \Return$,
  then also $m_x \succeq \Return$.
\end{lemma}
\begin{proof}
  $\forcing E$ may be of the form $L$ or $\plug E {\plug {\forcing F} L}$

  In the case of binding contexts $L$ we have $m_x = \Return$ by construction.

  In the case with a forcing frame,
  ${\forcing E} = \plug E {\plug {\forcing F} L}$, let us call $m_E$
  the mode of the hole of $E$. It is immediate that the mode imposed
  by $L$ on its hole is $\Return$, and that the mode imposed by
  $\forcing F$ on its own hole is $\Dereference$, so the total mode
  $m_x$ is $\mcomp {m_E} {\mcomp \Dereference \Return}$. We can prove by
  an easy induction on $E$ that $m_E$ is not $\Ignore$ or $\Delay$ --
  those are not evaluation contexts, so we have $m_E \succeq \Guard$.
  We conclude by monotonicity of mode composition:
  \begin{mathpar}
    m_x
    = \mcomp {m_E} {\mcomp \Dereference \Return}
    \succeq \mcomp \Guard {\mcomp \Dereference \Return}
    = \Dereference
  \end{mathpar}
\end{proof}

\begin{theorem}[\outlineref{lem:vicious}]
  $\der \emptyset t \Return$ never holds for $t \in \set{Vicious}$.
\end{theorem}

\begin{proof}
  Given $\der \empty t \Return$, let us assume that $t$ is $\plug E x$
  with no value binding for $x$ in $E$, and show that $E$ is not
  a forcing context.

  We implicitly assume that all terms are well-scoped, so the absence
  of value binding means that $x$ occurs in a \code{let rec} binding
  still being evaluated somewhere in $E$: $\plug E x$ is of the
  form \newcommand{\Ein}{E_\mathsf{in}}
  \newcommand{\Eout}{E_\mathsf{out}}
  \newcommand{\trec}{t_\mathsf{rec}}
  \begin{mathpar}
    \plug E x = \plug \Eout \trec

    \trec = (\letrecin {b, y = {\plug \Ein x}, b'} u)
  \end{mathpar}
  where $x$ is bound in $b$, $b'$ or is $y$ itself.

  Given our \code{let rec} typing rule (see \myref{Figure}{fig:rules}),
  the typing derivation for $t$ contains a sub-derivation for $\trec$ of the form
  \begin{mathpar}
    \infer
    {\fam i
      {\der
        {\Gamma_i, \fam j {\of {x_j} {m_{i,j}}}}
        {t_i}
        \Return}
     \\
     \fam {i,j} {m_{i,j} \preceq \Guard}
     \\\\
     \fam i {\Gamma'_i =
       \envsum {\Gamma_i} {\envbigsum {\fam j {\mcomp {m_{i,j}} {\Gamma'_j}}}}}
   }
    {\derbinding
      {\fam i {\of {x_i} {\Gamma'_i}}}
      {\fam i {x_i = t_i}}}
  \end{mathpar}
  In particular, the premise for $\plug \Ein x$ is of the form
  $\der {\Gamma, \fam j {\of {x_j} {m_j}}} {\plug \Ein x} \Return$ with
  $\fam j {x_j \preceq \Guard}$,
  and in particular $x \preceq \Guard$
  so $x \nsucceq \Return$.

  By \myfullref{Lemma}{lem:forcing-modes},
  $\Ein$ cannot be a forcing context,
  and in consequence $E$ is not forcing either.
\end{proof}

\begin{theorem}[\outlineref{thm:subject-reduction}]
  If $\der \Gamma t m$ and $\rew t {t'}$ then $\der \Gamma {t'} m$.
\end{theorem}

\begin{proof}
  We reason by inversion on the typing derivation of redexes, first
  for head-reduction $\rewhead t {t'}$ and then for reduction
  $\rew t t'$.

  \paragraph{Head reduction} We only show the head-reduction case for functions;
  pattern-matching is very similar. We have:
  \begin{mathpar}
  \infer*
  {\infer*
    {\der {\Gamma_t, \of x {m_x}} t {\mcomp {\mcomp m \Dereference} \Delay}}
    {\infer
      {\der {\Gamma_t} {\lam x t} {\mcomp m \Dereference}}
      {\der {\Gamma_t} {\plug L {\lam x t}} {\mcomp m \Dereference}}} \\
   \der {\Gamma_v} v {\mcomp m \Dereference}}
  {\der {\envsum {\Gamma_t} {\Gamma_v}} {\app {\plug L {\lam x t}} v} m}
  \end{mathpar}

  By associativity, $\mcomp {\mcomp m \Dereference} \Delay$ is
  the same as $\mcomp m \Dereference$.

  By subsumption,
  $\der {\Gamma_t, \of x {m_x}} t {\mcomp m \Dereference}$
  implies
  $\der {\Gamma_t, \of x {m_x}} t m$.

  To conclude by using \myfullref{Lemma}{lem:substitution},
  we must reconcile the mode of the argument
  $v : \mcomp m \Dereference$ with the (apparently arbitrary) mode
  $x : m_x$ of the variable. We reason by an inelegant case
  distinction.
  \begin{itemize}
  \item If $\mcomp m \Dereference$ is $\Dereference$, then by
    inversion (\myref{Lemma}{lem:invert-dereference}) either $m_x$ is
    $\Dereference$ (problem solved) or $x$ does not occur in $t$ (no
    need for the substitution lemma).
  \item If $\mcomp m \Dereference$ is not $\Dereference$, then $m$
    must be $\Ignore$ or $\Delay$. If it is $\Ignore$, inversion
    (\myref{Lemma}{lem:invert-ignore}) directly proves our goal. If it is
    $\Delay$, then by inversion (\myref{Lemma}{lem:invert-delay}) $m_x$
    itself can be weakened (subsuming the derivation of $t$) to be
    below $\Delay$.
  \end{itemize}

  \paragraph{Reduction under context} Reducing a head-redex under
  context preserves typability by the argument above. Let us consider
  the lookup case.
  \begin{mathpar}
    \infer
    {\inctx x v E}
    {\rew {\plug E x} {\plug E v}}
  \end{mathpar}
  By inspecting the $\inctx x v E$ derivation, we find a value binding
  $B$ within $E$ with $x = v$, and a derivation of the form
  \begin{mathpar}
    \infer
    {\derbinding {\fam i {\of {x_i} {\Gamma'_i}}} B
      \\
      \fam i {m'_i} \defeq \fam i {\max(m_i, \Guard)}
      \\
      \der {\Gamma_u, \fam i {\of {x_i} {m_i}}} u m}
    {\der
      {\envsum {\envbigsum {\fam i {\mcomp {m'_i} {\Gamma'_i}}}} {\Gamma_u}}
      {\letrecin B u}
      m}

    \infer
    {\fam i
      {\der
        {\Gamma_i, \fam j {\of {x_j} {m_{i,j}}}}
        {v_i}
        \Return}
     \\
     \fam {i,j} {m_{i,j} \preceq \Guard}
     \\\\
     \fam i {\Gamma'_i =
       \envsum {\Gamma_i} {\envbigsum {\fam j {\mcomp {m_{i,j}} {\Gamma'_j}}}}}
    }
    {\derbinding
      {\fam i {\of {x_i} {\Gamma_i}}}
      {\fam i {x_i = v_i}}}
  \end{mathpar}
  By abuse of notation, we will write $m_x$, $\Gamma_x$ and
  $\Gamma'_x$ to express the $m_i$, $\Gamma_i$ and $\Gamma'_i$ for the
  $i$ such that $x_i = x$.

  The occurrence of $x$ in the hole of $\plug E \hole$ is typed
  (eventually by a variable rule) at some mode $m_\hole$. The
  declaration-side mode $m_x$ was built by collecting the usage modes
  of all occurrences of $x$ in the \code{let rec} body $u$, which in
  particular contains the hole of $E$, so we have
  $m_\hole \preceq m_x$ by \myfullref{Lemma}{lem:environment-flow}.

  The binding derivation gives us a proof
  $\der {\Gamma_x, \Gamma_{\kwd{rec}}} v \Return$ that the binding
  $x = v$ was correct at its definition site, where
  $\Gamma_{\kwd{rec}}$ has exactly the mutually-recursive variables
  $\fam i {\of {x_i} {m_i}}$. Notice that this subderivation is
  completely independent of the ambient expected mode $m$.

  By \myfullref{Theorem}{thm:localization}, we can compose
  this within $m_\hole$ to get a derivation
  $\der {\mcomp {m_\hole} {\Gamma_x, \Gamma_{\kwd{rec}}}}
  v {m_\hole}$, that we wish to substitute into the hole of $E$. First
  we weaken it (\myref{Lemma}{lem:weakening}) into the judgment
  $\der {\mcomp {m_x} {\Gamma_x, \Gamma_{\kwd{rec}}}} v {m_\hole}$.

  Plugging this derivation in the hole of $E$ requires weakening the
  derivation of $u$ (the part of $\plug E \hole$ that is after the
  declaration of $x$) to add the environment
  $\mcomp {m_x} {\Gamma_x, \Gamma_{\kwd{rec}}}$. Weakening is
  always possible (\myref{Lemma}{lem:weakening}), but it may change the
  environment of the derivation, while we need to preserve the
  environment of $\plug E x$. Consider the following valid derivation:
  \begin{mathpar}
    \infer
    {\derbinding {\fam i {\of {x_i} {\Gamma'_i}}} B
      \\
      \fam i {m''_i} \defeq
      \fam i {\max(\max(m_i, \mcomp{m_x}{\Gamma_{\kwd{rec}}}(x_i)), \Guard)}
      \\
      \der {\envsum {\Gamma_u} {\mcomp {m_x} {\Gamma_x}},
            \envsum {\fam i {\of {x_i} {m_i}}}
              {\mcomp {m_x} {\Gamma_{\kwd{rec}}}}}
           {\subst u {\by v x}} m}
    {\der
      {\envsum
        {\envbigsum {\fam i {\mcomp {m''_i} {\Gamma'_i}}}}
        {\envsum{\Gamma_u}{\mcomp {m_x} {\Gamma_x}}}}
      {\letrecin B {\subst u {\by v x}}}
      m}
  \end{mathpar}
  To show that we preserve the environment of $\plug E x$, we show that this
  derivation is not in a bigger environment than the environment of our source term:
  \begin{mathpar}
    \envsum
        {\envbigsum {\fam i {\mcomp {m''_i} {\Gamma'_i}}}}
        {\envsum{\Gamma_u}{\mcomp {m_x} {\Gamma_x}}}

    \preceq

    \envsum {\envbigsum {\fam i {\mcomp {m'_i} {\Gamma'_i}}}} {\Gamma_u}
  \end{mathpar}

  By construction we have
  $m_x \preceq m'_x \preceq m''_x$ and
  $\Gamma_x \preceq \Gamma'_x$, so
  $\mcomp {m_x} {\Gamma_x} \preceq \mcomp {m''_x} {\Gamma'_x}$
  which implies
  \begin{mathpar}
    \envsum
        {\envbigsum {\fam i {\mcomp {m''_i} {\Gamma'_i}}}}
        {\envsum{\Gamma_u}{\mcomp {m_x} {\Gamma_x}}}

    \preceq

    \envsum
        {\envbigsum {\fam i {\mcomp {m''_i} {\Gamma'_i}}}}
        {\Gamma_u}
  \end{mathpar}

  Then, notice that $\Gamma_{\kwd{rec}}(x_i)$ is exactly $m_{x,i}$, so
  $m ''_i$ is $\max(m'_i, \mcomp {m_x} {m_{x,i}})$. We can thus rewrite
  $\mcomp {m''_i} {\Gamma'_i}$ into
  $\envsum {\mcomp {m'_i} \Gamma'_i} {\mcomp {\mcomp {m_x} {m_{x_i}}} {\Gamma'_i}}$,
  which gives
  \begin{mathpar}
    \envsum
    {\envbigsum {\fam i {\mcomp {m''_i} {\Gamma'_i}}}}
    {\Gamma_u}

    =

    \envsum
    {\envsum
      {\envbigsum {\fam i {\mcomp {m'_i} {\Gamma'_i}}}}
      {\mcomp {m_x} {\envbigsum {\fam i {\mcomp {m_{x,i}} {\Gamma'_i}}}}}
    }
    {\Gamma_u}
  \end{mathpar}

  The extra term
  ${\envbigsum {\fam i {\mcomp {m_{x,i}} {\Gamma'_i}}}}$ is precisely
  the term that appears in the definition of $\Gamma'_x$ from the
  $\fam i {\Gamma_i}$, taking into account transitive mutual dependencies --
  indeed, when we replace $x$ by its value $v$, we replace transitive
  dependencies on its mutual variables by direct dependencies on
  occurrences in $v$. We thus have
  \begin{mathpar}
    \envbigsum {\fam i {\mcomp {m_{x,i}} {\Gamma'_i}}}
    \preceq
    \Gamma'_x
  \end{mathpar}
  and can conclude with
  \begin{mathpar}
    \begin{array}{ll}
      &
      \displaystyle
      \envsum
      {\envsum
        {\envbigsum {\fam i {\mcomp {m'_i} {\Gamma'_i}}}}
        {\mcomp {m_x} {\envbigsum {\fam i {\mcomp {m_{x,i}} {\Gamma'_i}}}}}
      }
      {\Gamma_u}
      \\ \preceq &
      \displaystyle
      \envsum
      {\envsum
        {\envbigsum {\fam i {\mcomp {m'_i} {\Gamma'_i}}}}
        {\mcomp {m_x} {\Gamma'_x}}
      }
      {\Gamma_u}
      \\ \preceq &
      \displaystyle
      \envsum
      {\envbigsum {\fam i {\mcomp {m'_i} {\Gamma'_i}}}}
      {\Gamma_u}
    \end{array}
  \end{mathpar}
\end{proof}

\section{Compilation to Global Store: Simulation Proof}
\label{ann:simulation-proof}

This section explains the definition of the relation in
Figure~\ref{fig:head-term-relation}, and develops the simulation
result to prove \myfullref{Theorem}{thm:backward-simulation} and
\myfullref{Theorem}{thm:no-segfault}.

\begin{remark}[$\alpha$-equivalence of configurations]
  In a configuration $\conf H t$, the heap $H$ binds the variables of
  its domain in $t$. In particular, consistently renaming a variable
  in $H$ and in $t$ results in an $\alpha$-equivalent configuration;
  for example $\conf {\heapext {} x \bot} x$ and
  $\conf {\heapext {} y \bot} y$ are $\alpha$-equivalent. Our
  operations on configurations respect
  $\alpha$-equivalence.
\end{remark}

\begin{fact}[Congruence]\label{fact:congruence}
  If
  $\rew {\conf H t} {\conf {H'} {t'}}$ then
  $\rew
    {\conf {\heapjoin {H_E} H} {\plug E t}}
    {\conf {\heapjoin {H_E} {H'}} {\plug E {t'}}}$.
\end{fact}

\begin{fact}[Relevant head reduction]\label{fact:relevant-head-reduction}
  If $\rewhead {\conf H t} {\conf {H'} {t'}}$ then the domain of $H$
  contains only free variables of $t$; furthermore, $t'$ and the domain of $H'$
  contain only free or bound variables of $t$.
\end{fact}

\subsection{Relating the Local-Store Term with Global-Store Configurations}

\subsubsection{Simple rules} The simple rules relate term-formers that
exist in both the source and target language; they simply ask the
evaluable subterms to be related and take the union of the
corresponding heaps.

\subsubsection{$\kwd{letrec}$ rules}
The $\kwd{letrec}$ rules relate the $\kwd{letrec}$ construct in the source
language to store-update operations in the target. There are several
rules, that correspond to distinct ``moments'' in the computation of
the translation of a $\letrecin {\fam {i \in I} {x_i = t_i}} u$ binding, which
initially gets translated as
  $\newin {\fam i {x_i}} {\seq {\para {\fam i {\setref {x_i} {t_i}}}} {\comp u}}$:
\begin{itemize}
\item In the first moment we are reducing the allocations
  $\newin {\fam i {x_i}} {}$: some of them have already been allocated in the heap,
  others are yet to be evaluated. The rest of the target term is not in reducible
  position, so it is unchanged. This is the rule \Rule{init}.
\item In the second moment, all store locations have been created, and
  we are evaluating the store updates
  $\fam i {\setref {x_i} {t_i}}$. Some of them have been fully reduced
  to values and then written in the store, others are still
  (partially evaluated) terms. This is the rule \Rule{write}.
\item In the third moment, the store updates have been performed
  so the $\para \dots$ block reduced to $\Done$. This is the rule \Rule{done}.
\item In the fourth and last moment, we are in the process of further
  evaluating the body of the $\kwd{letrec}$-definition, $u$; the target is
  not restricted to the source translation $\comp u$ anymore, it may
  be any target configuration related to $u$. This is the rule
  \Rule{further}.
\end{itemize}

\begin{example}
  The following relations hold:
  \begin{mathpar}
    \termrel
      {\app x {(\letrecin {y = {\cS y}} y)}}
      {\conf \emptyset {\app x {(\newin y {\seq {\para {(\setref y {\cS y})}} y})}}}

    \termrel
      {\app x {(\letrecin {y = {\cS y}} y)}}
      {\conf {\heapext {} y \bot} {\app x {(\seq {\para {(\setref y {\cS y})}} y)}}}
    \end{mathpar}
    \begin{mathpar}
    \termrel
      {\app x {(\letrecin {y = {\cS y}} y)}}
      {\conf {\heapext {} y {\cS y}} {\app x y}}

    \termrel
      {(\letrecin {y = {\cS y}} {\app x y})}
      {\conf {\heapext {} y {\cS y}} {\app x y}}
  \end{mathpar}
\end{example}

\subsubsection{Heap rules} Finally, the heap rules give more to relate
local heaps and global heaps (than simple disjoint union of all
local heaps), to account for the fact that local heaps (explicit
substitutions) may be duplicated or erased during the reduction of our
source terms, as demonstrated in our Example~\ref{ex:store-duplication}.

When trying to prove that a source term is related to a target
configuration, the \Rule{heap-weaken} rule allows let us discard
a fully-evaluated part of the global target store that is not
referenced from the rest of the heap or the target term. This lets us
ignore, when relating two terms, the parts of the global store that
correspond to local stores that have been erased by $\beta$-reduction,
for example when reducing the source term
$\app {(\lam x \Done)} {(\letrecin {x = \kwd{Foo}} x)}$ and its
related target term.

\begin{remark}
  The restriction that only value heaps $B$ may be weakened, instead
  of arbitrary heap fragments $H'$, gives us a stronger inversion
  principle in Fact~\ref{fact:target-value-inversion}: in certain
  cases we can prove that the heap of a related configuration must be
  a value heap, which would never be the case if one could always
  weaken arbitrary heaps.
\end{remark}

The \Rule{heap-copy} rule let us relate a source term and a target
configuration even when the source term contains several copies of
a local store fragment that is included only once in the target global
store: by this rule it suffices to prove the relation on a larger
global store that includes several copies of some parts of the
original store.

The rule is defined for any \emph{renaming} $\phi$ from heap locations
(variables) to locations, that may ``collapse'' several locations from
the larger heap $H$ into the same location in the resulting heap
$\phi(H)$. To apply this rule we require that $\phi$ be compatible
with the heap $H$, which intuitively means that locations $x, y$ in
$H$ that are collapsed by $\phi$ into in the same location must have
pointed to the same value after $\phi$-renaming. We also require
$\phi$ be the identity on variables outside $\dom H$; without this
restriction, applying $\phi$ could change the free variables of the
target term.

% Note to self: This domain condition is not included in the
% definition of "compatible with H", because it is global rather than
% local. Currently we can prove compatibility piecewise on fragments
% of the heap, using the "Union compatibility" Fact, which is very
% convenient. Asking for "phi compatible with H" that phi be the
% identity outside H would break this modularity.

\begin{definition}[Finite renaming]\label{def:finite-renaming}
  A (finite) \emph{renaming} $\phi$ is a total function from variables
  to variables whose domain $\dom \phi$, defined as the set
  $\{ x \mid \phi(x) \neq x \}$, is finite.

  We write $\phi(t)$ for the replacement of each free variable
  $x \in t$ by $\phi(x)$.
\end{definition}

\begin{definition}[Heap compatibility]\label{def:heap-compatibility}
  We write $(\compatible \phi H)$ if the following properties hold:
  \begin{description}
  \item[functionality]
    $\forall x, y \in \dom H,\quad\phi(x) = \phi(y) \implies \phi(H(x)) = \phi(H(y))$
  \item[definedness]
    $\forall x, y \in \dom H,\quad\phi(x) = \phi(y) \wedge H(x) = \bot \implies x = y$
  \end{description}
\end{definition}

Functionality implies that
$\phi(H) \defeq \{ \phi(x) \mapsto \phi(t) \mid (x \mapsto t) \in H \}$
is well-defined as a finite map.

Definedness implies that $\phi$ is injective on the uninitialized
locations of $H$: distinct uninitialized locations cannot be collapsed
together.

\begin{example}
  The following relations hold:

  \begin{mathpar}
    \begin{array}{l}
      {\app {\app x
          {(\letrecin {y = {\cS y}} y)}}
        {(\letrecin {y' = {\cS {y'}}} {y'})}}
      \\
      \termrel {}
      {\conf
        {\heapext {\heapext {}
            y {\cS y}}
          {y'} {\cS {y'}}
        } {\app {\app x y} {y'}}}
    \end{array}

    \begin{array}{l}
      {\app {\app x
          {(\letrecin {y = {\cS y}} y)}}
        {(\letrecin {y' = {\cS {y'}}} {y'})}}
      \\
      \termrel {}
      {\conf
        {\heapext {}
            y {\cS y}}
        {\app {\app x y} {y}}}
    \end{array}
    \end{mathpar}

    The second relation is proved by using the \Rule{heap-copy} rule on the first one,
    with the heap function $\phi$ defined by $\phi(y') = \phi(y) = y$ and,
    for any $z \neq y'$, $\phi(z) = z$. This $\phi$ is compatible with the example heap
    because $\phi(\cS y) = \phi(\cS {y'}) = \cS y$.
\end{example}

\begin{fact}[Union compatibility]\label{fact:union-compatibility}
  If $\phi$ is compatible with $H_1$ and $H_2$,
  and $\phi(H_1)$ and $\phi(H_2)$ have disjoint domains,
  then $\phi$ is compatible with $\heapjoin {H_1} {H_2}$.
\end{fact}

\paragraph{Properties of the relation}
The simple rules are syntax-directed. The others are not:
\begin{itemize}
\item The $\kwd{letrec}$-binding rules relate the same source term
  to several different target terms.

  In the case of the rule \Rule{further}, the target term is exactly
  the one of a premise: any term $u_t$ may be related to a term of the
  form $\letrecin B \dots$ if it is the result of a reduction sequence
  that introduced $B$ in the global heap.

\item The heap rules act only on the target heap, not on the source or
  target terms. In particular, they may occur at the conclusion of any
  relation derivation.
\end{itemize}

\begin{fact}\label{fact:related-variables}
  A related term and configuration have the same free variables.
\end{fact}

\begin{fact}\label{fact:translation-related}
  For any source term $t$ we have $\termrel t {\conf \emptyset {\comp t}}$.
\end{fact}
\begin{proof}
  By immediate induction on $t$. In the $\kwd{letrec}$ case, use the
  \Rule{init} rule with $I = \emptyset$ --- no location initialized
  yet. For all other term-formers, use the simple rules.
\end{proof}

\begin{fact}[Target value inversion]\label{fact:target-value-inversion}
  If we have $\termrel t {\conf H {v_t}}$, a relation where the target term
  is a value, then:
  \begin{itemize}
  \item $t$ is a source value of the form $\plug L {v_s}$,
  \item $H$ is a value heap $B$,
  \item $\termrel {v_s} {\conf \emptyset {v_t}}$ are related by
    a simple rule in the derivation of $\termrel {\plug L {v_s}} {\conf B {v_t}}$.
  \end{itemize}
\end{fact}
\begin{proof}
  By induction on the relation derivation.

  The property is directly satisfied by the simple rules that may
  related a target value ($\lambda$-abstraction and data constructor),
  with $L = \square$ and $B = \emptyset$.

  The only letrec-related rule whose target may be a value is
  \Rule{further}, and it preserves this property (growing $L$ and $B$
  with the recursive bindings).

  The heap rules also preserve this property; they may only transform
  the value heap $B$ into another value heap. This is by definition
  for \Rule{heap-weaken}: it is restricted to weakening value heaps
  only. In the case of \Rule{heap-copy}, remark that if $\phi(H)$ is
  a value heap, then $H$ must also be a value heap: $\phi$ is
  a variable renaming, so we can only have $\phi(v?) = \bot$ if
  $v? = \bot$.
\end{proof}

Note that we do not have such an inversion principle for source
values: if we have $\termrel {v_s} {\conf H t}$, then the target
source term $t$ need not be a value, as $v_s$ may contain value
bindings that are not yet evaluated in $t$. For example we have
\begin{mathpar}
 \termrel
   {\letrecin {y = \cS y} {\cS y}}
   {\conf {\heapext {} y \bot} {\seq {\para {(\setref y {\cS y})}} {\cS y}}}
\end{mathpar}
where the source term is a value.

\subsection{Backward Simulation}

\begin{lemma}[Relation substitution]\label{lem:relation-substitution}~\\
If
\begin{mathline}
  \termrel {t_s} {\conf H {t_t}}

  \termrel {u_s} {\conf B {u_t}}

  H, B \text{~disjoint}
\end{mathline}
then
\begin{mathline}
  \termrel
    {\subst {t_s} {\by {u_s} x}}
    {\conf
      {\heapjoin H B}
      {\subst {t_t} {\by {u_t} x}}}
\end{mathline}
\end{lemma}

\begin{proof}
  If $x$ does not occur in $t_t$, then this is exactly the \Rule{heap-weaken} rule.
  Let us now consider the case where there is at least one occurrence of $x \in {t_t}$.
  Let $\fam {i \in I} {x_i}$ be the (non-empty) family of occurrences of $x$ in $t_t$.

  Let us choose, for each $x_i$, an $\alpha$-equivalent copy
  $\conf {B_i} {u_{t,i}}$ of $\conf B {u_t}$, where $B_i$ is fresh for
  $\conf H {t_t}$.

  From $\termrel {t_s} {\conf H {t_t}}$ we get a derivation of
  \begin{mathline}
  \termrel
    {\subst {t_s} {\by {u_s} x}}
    {\conf
      {H \heapjoinfam {i \in I} {B_i}}
      {\substfam {t_t} {i \in I} {\by {u_{t,i}} {x_i}}}}
  \end{mathline}
  by replacing each sub-derivation $\termrel x {\conf \emptyset {x_i}}$
  by our assumptions $\termrel {u_s} {\conf {B_i} {u_{t,i}}}$.

  We now consider the map $\phi$ that maps each $y_k \in B_k$ for
  some $k$ to the corresponding $y \in B$, and all other variables
  to themselves. By this definition we have that:
  \begin{itemize}
  \item $\dom{\phi} \subseteq \heapjoinfam i {B_i}$;
  \item $\phi$ is the identity on $H$, so it is
    compatible with $H$;
  \item $\phi(B_i) = B$ and $\fam i {\phi({u_{t,i}}) = {u_t}}$;
  \item $\phi$ is functional on $\heapjoinfam i {B_i}$;
    definedness trivially holds on value heaps,
    so $\phi$ is compatible with $\heapjoinfam i {B_i}$.
  \end{itemize}

  Furthermore, $H$ and $\phi(\heapjoinfam i {B_i} = B$ have disjoint domains,
  so by Fact~\ref{fact:union-compatibility} we have that $\phi$ is compatible
  with $H \heapjoinfam i {B_i}$.

  As a consequence, from
  \begin{mathline}
  \termrel
    {\subst {t_s} {\by {u_s} x}}
    {\conf
      {H \heapjoinfam {i \in I} {B_i}}
      {\substfam {t_t} {i \in I} {\by {u_{t,i}} {x_i}}}}
  \end{mathline}
  we can apply the \Rule{heap-copy} rule with $\phi$ to deduce
  \begin{mathline}
  \termrel
    {\subst {t_s} {\by {u_s} x}}
    {\conf
      {\phi(H \heapjoinfam {i \in I} {B_i})}
      {\phi(\substfam {t_t} {i \in I} {\by {u_{t,i}} {x_i}})}}
  \end{mathline}
  that is our goal
  \begin{mathline}
  \termrel
    {\subst {t_s} {\by {u_s} x}}
    {\conf
      {\heapjoin H B}
      {\subst {t_t} {\by {u_t} x}}}
  \end{mathline}
\end{proof}

% Theorem [Backward Simulation]:
% > if
% >
% >     tₛ ~ (H, tₜ)
% >             |
% >             v
% >         (H', tₜ')
% >
% > then there exists tₛ' such that
% >
% >     tₛ
% >     |
% >     v?
% >     tₛ' ~ (H', tₜ')

\begin{theorem*}[\outlineref{thm:backward-simulation}]~\\
If
\begin{mathline}
  \termrel {t_s} {\conf H {t_t}}

  \rew {\conf H {t_t}} {\conf {H'} {t'_t}}
\end{mathline}
then $\exists t'_s$,
\begin{mathline}
  \rewopt {t_s} {t'_s}

  \termrel {t'_s} {\conf {H'} {t'_t}}
\end{mathline}
\end{theorem*}

\begin{proof}
We proceed by induction on the derivation of $\termrel {t_s} {\conf H {t_t}}$.

The ``simple rules'' are the easy cases, with the interesting cases
coming heap-management or letrec-related cases.

\subsubsection*{Simple rules}

The cases of variables and $\lambda$-abstractions are immediate, as no
reduction can occur. For application, matching and constructors, the
same reasoning works in all cases, so we will only consider
application.

\begin{mathline}
\infer
  {\termrel {f_s} {\conf {H_f} {f_t}}
   \\
   \termrel {u_s} {\conf {H_u} {u_t}}}
  {\termrel
    {\app {f_s} {u_s}}
    {\conf {\heapjoin {H_f} {H_u}} {\app {f_t} {u_t}}}}
\end{mathline}

We proceed by case analysis on the reduction
$\rew {\conf H {t_t}} {\conf {H'} {t'_t}}$ with
$t_t = \app {f_t} {u_t}$ and $H = \heapjoin {H_f} {H_u}$.
The reduction is a \Rule{ctx} rule of the form
\begin{mathline}
  \infer
  {\rewhead {\conf {H_\square} {t_h}} {\conf {H'_\square} {t'_h}}}
  {\rew
    {\conf {\heapjoin {H_E} {H_\square}} {\plug E {t_h}}}
    {\conf {\heapjoin {H_E} {H'_\square}} {\plug E {t'_h}}}}
\end{mathline}

Either $E$ is the empty context $\hole$, which is the ``head
reduction'' case, or $E$ is non-empty, the ``non-head'' case.

\paragraph{Head reduction}

The interesting case is when we have a head reduction in the target
term: $f_t$ is of the form $(\lam x {g_t})$, and the reduction
happens on the redex $\app {(\lam x {g_t})} {u_t}$.

In this case $u_t$ must be a value $v_t$, and $f_t$ is also a value
(a $\lambda$-abstraction), so by
\myfullref{Fact}{fact:target-value-inversion} we know that:
\begin{itemize}
\item $\conf {H_u} {u_t}$ is a value configuration $\conf {B_v} {v_t}$
  and $u_s$ is a related value:\\ $u_s = \termrel {v_s} {\conf {B_v} {v_t}}$.
\item $\conf {H_f} {f_t}$ is a value configuration
  $\conf {B_f} {\lam x {g_t}}$.
\end{itemize}

More precisely, we know that $f_s$ is of the form
$\plug L {\lam x {g_s}}$, and that the subderivation of
$\termrel {\lam x {g_s}} {\lam x {g_t}}$ within the derivation of
$\termrel {f_s} {\conf {H_f} {f_s}}$ uses a simple rule. This is
necessarily the $\lambda$-abstraction rule, so we have
$\comp {g_s} = g_t$.

Note that if we replace this subderivation
$\termrel {\lam x {g_s}} {\conf \emptyset {\lam x {\comp {g_s}}}}$
by a derivation of $\termrel {g_s} {\conf \emptyset {\comp {g_s}}}$
(Fact~\ref{fact:translation-related}), we obtain a derivation of
$\termrel {\plug L {g_s}} {\conf {B_f} {\comp {g_s}}}$.

Our initial hypotheses
\begin{mathline}
  \termrel {t_s}
  {\rew {\conf H {t_t}} {\conf {H'} {t'_t}}}
\end{mathline}
became
\begin{mathline}
\termrel
  {\app {\plug L {\lam x {g_s}}} {v_s}}
  {\rew
    {\conf {\heapjoin {B_f} {B_v}} {\app {(\lam x {\comp {g_s}})} {v_t}}}
    {\conf {\heapjoin {B_f} {B_v}} {\subst {\comp {g_s}} {\by {v_t} x}}}}
\end{mathline}

We conclude this case of the proof with
\begin{mathline}
\rew
  {\app {\plug L {\lam x {g_s}}} {v_s}}
  {\termrel
    {\plug L {\subst {g_s} {\by {v_s} x}}}
    {\conf {\heapjoin {B_f} {B_v}} {\subst {\comp {g_s}} {\by {v_t} x}}}}
\end{mathline}
where the last relation is obtained from
$\termrel {\plug L {g_s}} {\conf {B_f} {\comp {g_s}}}$
and $\termrel {v_s} {\conf {B_v} {v_t}}$
by \myfullref{Lemma}{lem:relation-substitution}.

\paragraph{Non-head reduction}

If $\app {f_t} {u_t}$ is not the main redex of the reduction, we want
to proceed by induction on $f_t$ or $u_t$, depending on where the
reduction happens. Our evaluation context $\plug E \hole$ must be of
the form $\app {\plug {E_f} \hole} {u_t}$ or
$\app {f_t} {\plug {E_u} \hole}$; the proof of both cases is symmetric,
so we only discuss the case $\app {\plug {E_f} \hole} {u_t}$.

\begin{mathline}
  \infer
  {\termrel {f_s} {\conf {H_f} {f_t}}
   \\
   \termrel {u_s} {\conf {H_u} {u_t}}}
  {\termrel
    {\app {f_s} {u_s}}
    {\conf {\heapjoin {H_f} {H_u}} {\app {f_t} {u_t}}}}
\quad
  \infer
  {\rewhead {\conf {H_\square} {t_h}} {\conf {H'_\square} {t'_h}}}
  {\rew
    {\conf {\heapjoin {H_E} {H_\square}} {\app {\plug {E_f} {t_h}} {u_t}}}
    {\conf {\heapjoin {H_E} {H'_\square}} {\app {\plug {E_f} {t'_h} {u_t}}}}}
\end{mathline}

To proceed by induction on our premise
$\termrel {f_s} {\conf {H_f} {f_t}}$ we need the fact that $H_{\hole}$
is a sub-heap of $H_f$ (rather than $H_u$). This is ``obviously'' the
case, but we found it subtle to justify precisely.

For any variable $x \in \dom{H_\hole}$, let us show that $x$ belongs
to the domain of $H_f$ and not $H_u$ --- the two heaps are disjoint. By
\myfullref{Fact}{fact:relevant-head-reduction}, we know that $H_\hole$
is included in the free variables of $t_h$, so in particular $x$ is
free in $f_t$. Remember that we assume that, in each syntactic object
(term, configuration, derivation...), bound variables are distinct
from free variables; this can always be chosen to be the case by
performing $\alpha$-renamings appropriately. In the derivation of
$\termrel {\app {f_s} {u_s}} {\conf {\heapjoin {H_f} {H_u}} {\app {f_t} {u_t}}}$,
if $x$ was not in $\dom{H_f}$ it would be free in $\conf {H_f} {f_t}$;
in particular, it could not be bound in $\conf {H_u} {u_t}$, so we
know $x \notin \dom {H_u}$.

We thus have that $H_f$ is of the form $H_{E_f} \uplus H_\hole$,
so we can use our induction hypothesis on
\begin{mathline}
\termrel
  {f_s}
  {\infer*
    {\rewhead {\conf {H_\hole} {t_h}} {\conf {H'_\hole} {t'_h}}}
    {\rew
      {\conf {H_f} {\plug {E_f} {t_h}}}
      {\conf {H_{E_f} \uplus {H'_\hole}} {\plug {E_f} {t'_h}}}}}
\end{mathline}
to get a $f'_s$ such that
\begin{mathline}
\rewopt
  {f_s}
  {\termrel
    {f'_s}
    {\conf {H_{E_f} \uplus {H'_\hole}} {\plug {E_f} {t'_h}}}}
\end{mathline}
which lets us conclude our goal with
\begin{mathline}
\rewopt
  {\app {f_s} {u_s}}
  {\termrel
    {\app {f'_s} {u_s}}
    {\conf {\heapjoin {H_{E_f} \uplus {H'_\hole}} {H_u}} {\app {\plug {E_f} {t'_h}} {u_t}}}}
\end{mathline}

\subsubsection*{Letrec rules}~

Simulating the letrec rules is fairly simple, because the relation was
precisely designed to correspond to the dynamic semantics of the
translation of $\kwd{let rec}$ in the target language, with four rules
\Rule{init}, \Rule{write}, \Rule{done}, \Rule{further} that correspond
to four possible states of reduction of this translation, with
reductions from one state only to itself or to the next state in the
list.

\paragraph{Rule} \Rule{init}
\begin{mathline}
  \infer
  {J \neq \emptyset}
  {
    {\letrecin
      {\fam {i \in I} {x_i = t_i}, \fam {j \in J} {x_j = t_j}}
      u}
    \\ \termrel {}
    {\conf {\heapextfam {} {i \in I} {x_i} \bot}
      {\newin {\fam {j \in J} {x_j}}
        {\seq {\para {\fam {k \in I \uplus J} {\setref {x_k} {t_k}}}} {\comp u}}}
    }
  }
\end{mathline}

Any reduction of this target term is of the form below, with
$J = 1 \uplus {J'}$ where $1$ is some singleton set $\{\star\}$:

\begin{mathline}
  {\conf {\heapextfam {} {i \in I} {x_i} \bot}
    {\newin {x_\star} {\newin {\fam {j \in {J'}} {x_j}}
      {\seq {\para {\fam {k \in I \uplus 1 \uplus {J'}} {\setref {x_k} {t_k}}}} {\comp u}}}}}

\rew {}
  {\conf {\heapextfam {} {i \in I \uplus 1} {x_i} \bot}
    {\newin {\fam {j \in {J'}} {x_j}}
      {\seq {\para {\fam {k \in I \uplus 1 \uplus {J'}} {\setref {x_k} {t_k}}}} {\comp u}}}}
\end{mathline}

If $J' \neq \emptyset$, we have $\termrel {t_s} {\conf {H'} {t'_t}}$
again an instance of the \Rule{init} rule. If $J' = \emptyset$ then
$\newin {\fam {j \in J'} {x_j}} t$ is just $t$, so our reduced term is
of the form
\begin{mathpar}
  {\conf {\heapextfam {} {i \in I \uplus 1} {x_i} \bot}
    {\seq {\para {\fam {k \in I \uplus 1} {\setref {x_k} {t_k}}}} {\comp u}}}
\end{mathpar}
which is related to the same source term by the \Rule{write} rule,
with an empty set of committed values $\fam {i \in I} {v_{s,i}, v_{t,i}}$.

\paragraph{Rule} \Rule{write}
\begin{mathline}
\infer
  {\fam {i \in I} {\termrel {v_{s,i}} {\conf {B_i} {v_{t,i}}}}
   \\
   \fam {j \in J} {\termrel {t_{s,j}} {\conf {H_j} {t_{t,j}}}}
  }
  {
   {\letrecin {\fam {i \in I} {x_i = v_{s,i}}, \fam {j \in J} {y_j = {t_{s,j}}}} u}
   \\ \termrel {}
    {\conf
      {\heapjoin
        {\heapextfam
          {\heapextfam {} {i \in I} {x_i} {v_{t,i}}}
          {j \in J} {y_j} \bot}
        {\heapjoin {\fam {i \in I} {B_i}} {\fam {j \in J} {H_j}}}}
      {\seq {\para {(\fam {i \in I} \Done, \fam {j \in J} {\setref {y_j} {t_{t,j}}})}}
        {\comp u}}
    }}
\end{mathline}

There are three possible kinds of reduction for a related target term of this form:
\begin{enumerate}
\item We may be doing a reduction within one of the $\fam {j \in J} {t_{t,j}}$.
\item If one of the $\fam {j \in J} {t_{t,j}}$ is a value $v_{t,j}$,
  the write $\setref {y_j} {v_{t,j}}$ may be committed.
\item If the set of uncommitted bindings $\fam {j \in J} {\setref {y_j} {t_{t,j}}}$ is empty,
the subterm $\para {\fam {i \in I} \Done}$ reduces to $\Done$.
\end{enumerate}

In the first case, we have
$\rew {\conf {H_j} {t_{t,j}}} {\conf {H'_j} {t'_{t,j}}}$.\footnote{It
  is not immediate that the part of the global heap that is modified
  is precisely $H_j$, but it the same reasoning that we detailed in
  the ``Simple rules'' case. We will omit discussing similar instances
  of this point in the rest of the proof.} By induction hypothesis we
get $\termrel {t'_{s,j}} {\conf {H'_j} {t'_{t,j}}}$, and can conclude
by using the \Rule{write} rule again.

In the second case, we have $J = 1 \uplus J'$ where $t_{t,\star}$ is
a value $v_{t,\star}$, and thus $H_\star$ a value heap $B_\star$ by
\myfullref{Fact}{fact:target-value-inversion}. The reduction is thus
of the form
\begin{mathline}
  {\hconf
    {\heapjoin
      {\heapextfam
        {\heapext
          {\heapextfam {} {i \in I} {x_i} {v_{t,i}}}
          {y_\star} \bot}
        {j \in J'} {y_j} \bot}
      {\heapjoin {\fam {i \in I} {B_i}}
        {\heapjoin {B_\star} {\fam {j \in {1 \uplus J'}} {H_j}}}}}
    {\seq {\para {(\fam {i \in I} \Done,
                   \fam {\star \in 1} {\setref {y_\star} {v_{t,\star}}},
                   \fam {j \in {J'}} {\setref {y_j} {t_{t,j}}})}}
      {\comp u}}}

\rew {}
  {\hconf
    {\heapjoin
      {\heapextfam
        {\heapext
          {\heapextfam {} {i \in I} {x_i} {v_{t,i}}}
          {y_\star} {v_{t,\star}}}
        {j \in J'} {y_j} \bot}
      {\heapjoin {\fam {i \in I} {B_i}}
        {\heapjoin {B_\star} {\fam {j \in J'} {H_j}}}}}
    {\seq {\para {(\fam {i \in I} \Done,
                   \fam {\star \in 1} \Done,
                   \fam {j \in {J'}} {\setref {y_j} {t_{t,j}}})}}
      {\comp u}}}
\end{mathline}
and we can use the \Rule{write} rule again (with $I+1$ and $J'$) to
relate to the same source term.

In the third case, the reduction is
\begin{mathline}
\rew
  {\conf
    {\heapjoin
      {\heapextfam {} {i \in I} {x_i} {v_{t,i}}}
      {\fam {i \in I} {B_i}}}
    {\seq {\para {\fam {i \in I} \Done}}
      {\comp u}}}
  {\conf
    {\heapjoin
      {\heapextfam {} {i \in I} {x_i} {v_{t,i}}}
      {\fam {i \in I} {B_i}}}
    {\seq \Done {\comp u}}}
\end{mathline}
and the reduced target term is related to the initial source term by
\Rule{done}.

\paragraph{Rule} \Rule{done}
\begin{mathline}
  \infer
  {\fam {i \in I} {\termrel {v_{s,i}} {\conf {B_i} {v_{t,i}}}}}
  {\termrel
    {\letrecin {\fam {i \in I} {x_i = v_{s,i}}} u}
    {\conf
      {\heapjoin
        {\heapextfam {} {i \in I} {x_i} {v_{t,i}}}
        {\fam {i \in I} {B_i}} H}
      {\seq \Done {\comp u}}
    }
  }
\end{mathline}

The only possible reduction of the target term is
\begin{mathline}
\rew
  {\conf
    {\heapjoin
      {\heapextfam {} {i \in I} {x_i} {v_{t,i}}}
      {\fam {i \in I} {B_i}} H}
    {\seq \Done {\comp u}}}
  {\conf
    {\heapjoin
      {\heapextfam {} {i \in I} {x_i} {v_{t,i}}}
      {\fam {i \in I} {B_i}} H}
    {\comp u}}
\end{mathline}
The resulting target term is related to the initial source term by the
rule \Rule{further}, using \myfullref{Fact}{fact:translation-related}
to relate $u$ and $\comp u$.

\paragraph{Rule} \Rule{further}
\begin{mathline}
  \infer
  {\fam {i \in I} {\termrel {v_{s,i}} {\conf {B_i} {v_{t,i}}}}
   \\
   \termrel {u_s} {\conf H {u_t}}
  }
  {\termrel
    {\letrecin {\fam {i \in I} {x_i = v_{s,i}}} {u_s}}
    {\conf
      {\heapjoin
        {\heapextfam {} {i \in I} {x_i} {v_{t,i}}}
        {\heapjoin {\fam {i \in I} {B_i}} H}}
      {u_t}
    }
  }
\end{mathline}

The only possible reduction of the target term comes from a reduction
$\rew {\conf H {u_t}} {\conf {H'} {u'_t}}$. The proof in this case is
by immediate induction hypothesis on $\termrel {u_s} {\conf H {u_t}}$.

\subsubsection*{Heap rules}

\paragraph{Rule} \Rule{heap-weaken}
\begin{mathline}
\infer
{\termrel {t_s} {\conf H {t_t}}
 \\
 \forall x \in \dom{B},\ x \notin \conf H {t_t}}
{\termrel {t_s}
  {\conf {\heapjoin H B} {t_t}}}
\end{mathline}

By \myfullref{Fact}{fact:relevant-head-reduction}, any head reduction
$\rewhead {\conf {H_\hole} {t_h}} {\conf {H'_\hole} {t_h}}$ has
$H_\hole$ included in the free variables of $t_t$, and $H'_\hole$ in
its free or bound variables. In particular, those two sub-heaps of
$\heapjoin H B$ can be assumed disjoint from $B$, so they are
sub-heaps of $H$, with $H = \heapjoin {H_E} {H_\hole}$ and
$H_E, H'_\hole, B$ disjoint.

We can conclude by induction: from
\begin{mathline}
\termrel {t_s} {\rew
  {\conf {\heapjoin {H_E} {H_\hole}} {\plug E {t_h}}}
  {\conf {\heapjoin {H_E} {H'_\hole}} {\plug E {t'_h}}}}
\end{mathline}
we have by induction a $t'_s$ such that
\begin{mathline}
\rewopt {t_s} {\termrel
  {t'_s}
  {\conf {\heapjoin {H_E} {H'_\hole}} {\plug E {t'_h}}}}
\end{mathline}
and we can conclude (as $H'_\hole$ is disjoint from $B$) with the
\Rule{heap-weaken} rule again ($H_E, H'_\hole, B$ are disjoint)
\begin{mathline}
\rewopt {t_s} {\infer*
  {\termrel {t'_s}
    {\conf {\heapjoin {H_E} {H'_\hole}} {\plug E {t'_h}}}}
  {\termrel {t'_s}
    {\conf {\heapjoin {\heapjoin {H_E} {H'_\hole}} B} {\plug E {t'_h}}}}}
\end{mathline}

\paragraph{Rule} \Rule{heap-copy}
\begin{mathline}
\infer
  {\termrel {t_s} {\conf H {t_t}}
   \\
   \compatible \phi H
   \\
   \dom \phi \subseteq \dom H
  }
  {\termrel {t_s} {\conf {\phi(H)} {\phi(t_t)}}}
\end{mathline}
This is the delicate case of proof, testing our definition of
$(\compatible \phi H)$.

Our reduction must be a \Rule{ctx} rule, so $\phi(t_t)$ is of the form
$\plug {E_0} {t_{h,0}}$. Variable renamings preserve the term
structure, so the preimage by $\phi$ of $\plug {E_0} {t_{h,0}}$ must
itself be of the form $\plug E {t_h}$. Without loss of generality, we
can thus assume that $\phi(t_t)$ is of the form
$\plug {\phi(E)} {\phi(t_h)}$.

We then reason by case analysis on the possible head reduction rules
involving $\phi(t_h)$. In each case, we will refine this
without-loss-of-generality reasoning by inverting $\phi$ on the
specific structure of the head reduction rule.

\paragraph{Rule case: pure reductions}

Without loss of generality, we can assume that any pure reduction
from $\phi{\conf H {t_t}}$ must be of the form
\begin{mathline}
\infer
  {\rewhead {\conf \emptyset {\phi(t_h)}} {\conf \emptyset {u'_h}}}
  {\rew
    {\conf {\phi(H)} {\plug {\phi(E)} {\phi(t_h)}}}
    {\conf {\phi(H)} {\plug {\phi(E)} {u'_h}}}}
\end{mathline}

Substitutions preserve reductions: if $t \to t'$ in the
$\lambda$-calculus then $t[\sigma] \to t'[\sigma]$ for any
variable-to-terms substitution $\sigma$. In the case of variable
renaming $\phi$ only, we have a converse property that if $\phi(t)$
is reducible, then $t$ itself is reducible (a variable-to-variable
substitution cannot introduce a redex). If $\phi(t)$
reduces to some $u'$, then $t$ reduces to some $t'$; as substitutions
preserve reductions we furthermore have that $\phi(t') = u'$.

From our assumption
$\rewhead {\conf \emptyset {\phi(t_h)}} {\conf \emptyset {u'_h}}$ we
thus have $t'_h$ such that $\phi(t'_h) = u'_h$ and
$\rewhead {\conf \emptyset {t_h}} {\conf \emptyset {t'_h}}$. We thus have
\begin{mathline}
\termrel {t_s} {\conf H {t_t}}
=
  {\infer*
    {\rewhead {\conf \emptyset {t_h}} {\conf \emptyset {t'_h}}}
    {\rew {\conf H {\plug E {t_h}}} {\conf H {\plug E {t'_h}}}}}
\end{mathline}
which gives by induction hypothesis a $t'_s$ such that
\begin{mathline}
\rewopt {t_s} {\termrel {t'_s} {\conf H {\plug E {t'_h}}}}
\end{mathline}
which lets us conclude by applying \Rule{heap-copy} again:
\begin{mathline}
\rewopt {t_s}
  {\infer*
    {\termrel {t'_s} {\conf H {\plug E {t'_h}}}}
    {\termrel {t'_s} {\conf {\phi(H)} {\plug {\phi(E)} {u'_h}}}}
  }
\end{mathline}

\paragraph{Rule case:} \Rule{new}

Without loss of generality, a \Rule{new} reduction of a configuration
in the image of a variable renaming $\phi$ must start from a term of
the form $\phi {\conf H {\plug E {\newin x u}}}$. The bound variable
$x$ can be assumed outside the domain of $H$ and the free variables of
$\plug E {\newin x u}$, as well as outside their image through
$\phi$. We have $\dom \phi \subseteq \dom H$, so we know that $x$ is
also outside the domain of $\phi$, which gives $\phi(x) = x$. Thus any
such reduction can be assumed to be of the form
\begin{mathline}
\termrel {t_s}
  {\infer*
    {\rewhead
      {\conf \emptyset {\phi(\newin x u)}}
      {\conf {\heapext {} x \bot} {\phi(u)}}}
    {\rew
      {\conf {\phi(H)} {\plug {\phi(E)} {\phi(\newin x u)}}}
      {\conf {\heapext {\phi(H)} x \bot} {\phi(\plug E u)}}}}
\end{mathline}

From our premise
\begin{mathline}
  \termrel {t_s} {\conf H {\plug E {\newin x u}}}
\end{mathline}
we have by induction hypothesis a $t'_s$ such that
\begin{mathline}
  \rewopt {t_s}
  {\termrel {t'_s} {\conf {\heapext H x \bot} {\plug E u}}}
\end{mathline}
and we can conclude with the \Rule{heap-copy} rule again
\begin{mathline}
  \rewopt {t_s}
  {\inferrule*
    {\termrel {t'_s} {\conf {\heapext H x \bot} {\plug E u}}
     \\
     \compatible \phi {\heapext H x \bot}}
    {\termrel {t'_s} {\conf {\phi(\heapext H x \bot)} {\phi(\plug E u)}} \hspace{12em}}}
\end{mathline}
where the compatibility of $\phi$ with $\heapext H x \bot$ is a direct consequence
of \myfullref{Fact}{fact:union-compatibility}.

\paragraph{Rule case:} \Rule{set}

Without loss of generality, a \Rule{set} reduction from a term in the
image of a variable renaming $\phi$ must be of the form
\begin{mathline}
\rew
  {\phi{\conf {\heapext H x \bot} {\plug E {\setref x v}}}}
  {\phi{\conf {\heapext H x v} {\plug E \Done}}}
\end{mathline}
We also have a premise
$
\termrel {t_s} {\conf {\heapext H x \bot} {\plug E {\setref x v}}}
$
where the target reduces as
\begin{mathline}
\termrel {t_s}
  {\rew
    {\conf {\heapext H x \bot} {\plug E {\setref x v}}}
    {\conf {\heapext H x v} {\plug E \Done}}}
\end{mathline}
which gives by induction hypothesis a $t'_s$ such that
\begin{mathline}
\rewopt {t_s}
  {\termrel
    {t'_s}
    {\conf {\heapext H x v} {\plug E \Done}}}
\end{mathline}

We can conclude by applying \Rule{heap-copy} again on this last relation,
\begin{mathline}
\rewopt {t_s}
  {\infer*
    {\termrel {t'_s} {\conf {\heapext H x v} {\plug E \Done}}}
    {\termrel {t'_s} {\phi{\conf {\heapext H x v} {\plug E \Done}}}}}
\end{mathline}
provided that $\phi$ is compatible with $\heapext H x v$.
It is (see \myref{Definition}{def:heap-compatibility}):
\begin{description}
\item[definedness] is preserved from $\heapext H x \bot$
\item[functionality] holds for any pair of references in $H$. For
  pairs containing $x$, we know by definedness of $\phi$ on
  $\heapext H x \bot$ that $z \neq x$ implies $\phi(z) \neq \phi(x)$,
  which implies functionality.
\end{description}

\begin{remark}
  This case is the reason why definedness was introduced in definition
  of compatible renamings. Functionality is required for the notation
  $\phi(H)$ to even make sense, but definedness is not an obvious
  requirement. Here is it essential for \Rule{set} to preserve
  functionality.
\end{remark}

\paragraph{Rule case:} \Rule{lookup}

Without loss of generality, a \Rule{lookup} reduction of a configuration
in the image of a variable renaming $\phi$ must start from a term of
the form
\begin{mathline}
\rew
  {\phi {\conf {\heapext H x v} {\plug E x}}}
  {\phi {\conf {\heapext H x v} {\plug E v}}}
\end{mathline}
we also have a premise
$
\termrel {t_s} {\conf {\heapext H x v} {\plug E x}}
$
so we have the reduction
\begin{mathline}
\termrel {t_s}
  {\rew
    {\conf {\heapext H x v} {\plug E x}}
    {\conf {\heapext H x v} {\plug E v}}}
\end{mathline}
and thus by induction hypothesis a $t'_s$ such that
\begin{mathline}
\rewopt {t_s}
  {\termrel {t'_s}
    {\conf {\heapext H x v} {\plug E v}}}
\end{mathline}
which lets us conclude by applying \Rule{heap-copy} again (we know by
assumption that $\phi$ is compatible with $\heapext H x v$)
\begin{mathline}
\rewopt {t_s}
  {\infer*
    {\termrel {t'_s}
      {\conf {\heapext H x v} {\plug E v}}}
    {\termrel {t'_s}
      {\phi {\conf {\heapext H x v} {\plug E v}}}}}
\end{mathline}
\end{proof}

\subsection{Viciousness and Segfaults}

Besides showing a tight correspondence between our local and global
store semantics, the backward simulation result is intended to prove
that this compilation strategy preserves safety: if our static
analysis accepts a source term, then its compilation will not fail due
to undefined values. To establish this, we relate the notions of
failures in the target language (forcing contexts and
$\mathsf{Segfault}$ terms from Figure~\ref{fig:target-syntax}) to
failures in the source language (forcing contexts and
$\mathsf{Vicious}$ from Figure~\ref{fig:failures}).

\begin{fact}[No value forcing]\label{fact:no-value-forcing}~\\
  A target value $v_t$ is never of the form $\plug {E_{f,t}} t$
  for any target forcing context $E_{f,t}$.
\end{fact}

\begin{lemma}[Target forcing inversion]\label{lem:target-forcing-inversion}~\\
If $t_s$ is related to a variable-forcing target configuration
\begin{mathline}
  \termrel {t_s} {\conf H {\plug {E_{f,t}} x}}
\end{mathline}
then
\begin{mathline}
  t_s = \plug {E_{f,s}} x
\end{mathline}
for some source forcing context $E_{f,s}$.
\end{lemma}

\begin{proof}
A target forcing context $E_{f,s}$ is defined
as the identity $\hole$ or the composition of an ordinary evaluation context
$E$ with a forcing frame $F_f$ (see Figure~\ref{fig:target-syntax}).

The proof is by induction on the relation; for each layer of structure
in the relation, it may be part of $E$, then we proceed by induction,
of the forcing frame $F_f$, then it is a base case.

\begin{itemize}
\item variable case: immediate.
\item abstraction case: impossible.
\item application case: by direct induction
  if part of $E$, and immediate if part of $F_f$.
\item constructor case: by direct induction.
\item match case: as the application case.
\item \Rule{init}: impossible as $\newin x \hole$
  cannot be the prefix of a forcing context.
\item \Rule{write}: the target term is a forcing context if the
  context goes into one of the $\fam j {\setref {y_j} {t_{t,j}}}$ being evaluated --
  $\comp u$ is not in reducible position. The corresponding source
  term $t_{s,j}$ is in reducible position in the source, so we proceed
  by induction.
\item \Rule{done}: impossible.
\item \Rule{further}: by direct induction.
\item \Rule{heap-weaken}, \Rule{heap-copy}: by direct induction.
\end{itemize}
\end{proof}

\begin{lemma}[Definedness relation]\label{lem:definedness-relation}~\\
If
\begin{mathline}
  \ctxrel {\plug {E_s} \hole} {\conf H {\plug {E_t} \hole}}

  (z = v) \in {E_s}
\end{mathline}
then $H(z)$ is defined and distinct from $\bot$.
\end{lemma}

\begin{remark}
  By $
    \ctxrel {\plug {E_s} \hole} {\conf H {\plug {E_t} \hole}}
  $ we mean the relation between contexts that extends $(\termrel{}{})$
  with the rule $\ctxrel \hole {\conf \emptyset \hole}$, or equivalently
  $\termrel {\plug {E_s} x} {\conf H {\plug {E_t} x}}$ for some variable
  $x$ fresh for $E_s, H, E_t$.
\end{remark}

\begin{proof}
  The proof is by induction on the derivation of $
    \ctxrel {\plug {E_s} \hole} {\conf H {\plug {E_t} \hole}}
  $. Most cases are immediate as they do not add bindings to the context.
  \begin{itemize}
  \item hole case $\hole$: immediate (there is no $z$ such that $(z = v) \in \hole$).
  \item variable case: impossible (not of the form $\plug E \hole$).
  \item application case: the holes must be on the same side in the
    source and target applications (as a hole is only related to
    a hole), so we can proceed by direct induction.
  \item construction case: similar to the application case.
  \item match case: by direct induction.
  \item \Rule{init}, \Rule{done}: impossible (the target is not of the form
    $\plug E \hole$).
  \item \Rule{write} and \Rule{further}: these are the interesting
    cases, proved below.
  \item \Rule{heap-weaken}, \Rule{heap-copy}: by direct induction.
  \end{itemize}

\paragraph{Rule} \Rule{write}
\begin{mathline}
  \infer
  {\fam {i \in I} {\termrel {v_{s,i}} {\conf {B_i} {v_{t,i}}}}
   \\
   \fam {j \in J} {\termrel {t_{s,j}} {\conf {H_j} {t_{t,j}}}}
  }
  {
   {\letrecin {\fam {i \in I} {x_i = v_{s,i}}, \fam {j \in J} {y_j = {t_{s,j}}}} u}
   \\ \termrel {}
    {\conf
      {\heapjoin
        {\heapextfam
          {\heapextfam {} {i \in I} {x_i} {v_{t,i}}}
          {j \in J} {y_j} \bot}
        {\heapjoin {\fam {i \in I} {B_i}} {\fam {j \in J} {H_j}}}}
      {\seq {\para {(\fam {i \in I} \Done, \fam {j \in J} {\setref {y_j} {t_{t,j}}})}} {\comp u}}
    }}
\end{mathline}

In the target term, $\comp u$ is not in reducible position. The only
related reducible subterms are the $\fam {j \in J} {t_{t,j}}$ in the
target, which correspond to an evaluation context frame $F$ of the
form $\letrecin {b, x = \hole, b'} u$ in the source. We never have
$(x = v) \in F$ for this partly-evaluated frame, so $z$ must come from
the corresponding $H_j$ and we proceed by induction on the
corresponding premise $\termrel {t_{s,j}} {\conf {H_j} {t_{t,j}}}$.

\paragraph{Rule} \Rule{further}
\begin{mathline}
  \infer
  {\fam {i \in I} {\termrel {v_{s,i}} {\conf {B_i} {v_{t,i}}}}
   \\
   \termrel {u_s} {\conf H {u_t}}
  }
  {\termrel
    {\letrecin {\fam {i \in I} {x_i = v_{s,i}}} {u_s}}
    {\conf
      {\heapjoin
        {\heapextfam {} {i \in I} {x_i} {v_{t,i}}}
        {\heapjoin {\fam {i \in I} {B_i}} H}}
      {u_t}
    }
  }
\end{mathline}
If $z$ is one of the $\fam {i \in I} x_i$, we have $\heapext {} z {v_{t,i}}$ in the
target heap so $H(z) \neq \bot$ as claimed. Otherwise $z$ must be in a $B_I$ or $H$,
and we proceed by induction on the corresponding premise.
\end{proof}

\begin{corollary}[Segfaults are Vicious]\label{cor:segfaults-are-vicious}
  \begin{mathline}
    \termrel {t_s} {\conf H {t_t}}

    \wedge

    \conf H {t_t} \in \mathsf{Segfault}

    \implies

    t_s \in \mathsf{Vicious}
  \end{mathline}
\end{corollary}

\begin{proof}
  $\conf H {t_t}$ is in $\mathsf{Segfault}$, so by definition it is of
  the form $\conf {\heapext {H'} x \bot} {\plug {E_{t,f}} x}$ for some
  target forcing context $E_{t,f}$
  (Figure~\ref{fig:target-syntax}). In particular, we have
  $H(x) = \bot$.

  By \myfullref{Lemma}{lem:target-forcing-inversion}, the related source term $t_s$
  must be of the form $\plug {E_{s,f}} x$ for some source forcing context $E_{s,f}$.

  We cannot have $(x = v) \in E_{s,f}$, as by
  \myfullref{Lemma}{lem:definedness-relation} this would imply
  $H(x) \neq \bot$. Therefore $t_s$ is in $\mathsf{Vicious}$
  (Figure~\ref{fig:failures}).
\end{proof}

\begin{theorem*}[\outlineref{thm:no-segfault}]
\begin{mathline}
  \der \emptyset {t_s} \Return

  \wedge

  \rewstar {\conf \emptyset {\comp  {t_s}}} {\conf H {t'_t}}

  \implies

  {\conf H {t'_t}} \notin \mathsf{Segfault}
\end{mathline}
\end{theorem*}

\begin{proof}
  The metatheory of our static analysis tells us that
  a $\Return$-typed closed program $t_s$ cannot reduce to a vicious term: if
  $\rewstar {t_s} {t'_s}$, then $t'_s \notin \mathsf{Vicious}$.

  We have $\termrel {t_s} {\comp \emptyset {t_t}}$; by
  \myfullref{Theorem}{thm:backward-simulation}, for any
  $\rewstar {\conf \emptyset {\comp {t_s}}} {t'_t}$ there is a $t'_s$
  such that $\rewstar {t_s} {t'_s}$ and
  $\termrel {t'_s} {\conf H {t'_t}}$. In particular,
  $\conf {H} {t'_t}$ cannot be in $\mathsf{Segfault}$, as then the
  related $t'_s$ would be in $\mathsf{Vicious}$ by
  \myfullref{Corollary}{cor:segfaults-are-vicious}.
\end{proof}

\renewcommand{\theTotPages}{29}
\end{document}